\documentclass[preprint,1p,11pt]{IR-Template/ISAS_IR}

\usepackage[english]{babel}
\usepackage[utf8]{inputenc}
\usepackage{calc}
\usepackage[cmex10]{amsmath}
\usepackage{amssymb}
\usepackage[ruled]{algorithm2e}
\usepackage{caption}
\usepackage{subcaption}
\usepackage{flushend}
\usepackage{epstopdf}
\usepackage[protrusion=false]{microtype}
\usepackage{amsthm}
\usepackage{booktabs} 
\usepackage{grffile} 
\usepackage{lmodern}

\usepackage[table]{xcolor}

\usepackage{ISASPackages/ISASMatheMacros}
\usepackage{gkmacros}

\newtheorem{definition}{Definition}
\newtheorem{lemma}{Lemma}
\newtheorem{theorem}{Theorem}
\newtheorem{remark}{Remark}

\DeclareMathOperator{\atan2}{atan2}
\DeclareMathOperator{\sign}{sign}

\begin{document}

\clubpenalty = 10000
\widowpenalty = 10000
\displaywidowpenalty = 10000
\brokenpenalty = 10000
\finalhyphendemerits = 10000

\begin{frontmatter}

\title{Recursive Bayesian Filtering\\ in Circular State Spaces}

\author[isas]{Gerhard~Kurz}
\ead{gerhard.kurz@kit.edu}

\author[isas]{Igor~Gilitschenski}
\ead{gilitschenski@kit.edu}

\author[isas]{Uwe~D.~Hanebeck}
\ead{uwe.hanebeck@ieee.org}

\address[isas]{Intelligent Sensor-Actuator-Systems Laboratory (ISAS)\\
Institute for Anthropomatics and Robotics\\
Karlsruhe Institute of Technology (KIT), Germany\vspace{3mm}}

\begin{abstract}
For recursive circular filtering based on circular statistics, we introduce a general framework for estimation of a circular state based on different circular distributions, specifically the wrapped normal distribution and the von Mises distribution. We propose an estimation method for circular systems with nonlinear system and measurement functions. This is achieved by relying on efficient deterministic sampling techniques. Furthermore, we show how the calculations can be simplified in a variety of important special cases, such as systems with additive noise as well as identity system or measurement functions. We introduce several novel key components, particularly a distribution-free prediction algorithm, a new and superior formula for the multiplication of wrapped normal densities, and the ability to deal with non-additive system noise. All proposed methods are thoroughly evaluated and compared to several state-of-the-art solutions.
\end{abstract}

\end{frontmatter}

\section{Introduction}
\label{sec:introduction}
Estimation of circular quantities is an omnipresent issue, be it the  wind direction, the angle of a robotic revolute joint, the orientation of a turntable, or the direction a vehicle is facing. Circular estimation is not limited to applications involving angles, however, and can be applied to a variety of periodic phenomena. For example phase estimation is a common issue in signal processing, and tracking objects that periodically move along a certain trajectory is also of interest.

Standard approaches to circular estimation are typically based on estimation techniques designed for linear scenarios that are tweaked to deal with some of the issues arising in the presence of circular quantities. However, modifying linear methods cannot only be tedious and error-prone, but also yields suboptimal results, because certain assumptions of these methods are violated. For example, solutions based on Kalman filters \cite{kalman1960} or nonlinear versions thereof \cite{julier2004} fundamentally neglect the true topology of the underlying manifold and assume a Gaussian distribution, which is only defined on $\mathbb{R}^n$. In the linear case, the use of a Gaussian distribution is frequently justified by the central limit theorem. This justification no longer holds in a circular setting, as the Gaussian is not a limit distribution on the circle.

In order to properly deal with circular estimation problems, we rely on circular statistics \cite{jammalamadaka2001}, \cite{batschelet1981}, a subfield of statistics that deals with circular quantities. More broadly, the field of directional statistics \cite{mardia1999} considers a variety of manifolds, such as the circle, the hypersphere, or the torus. Unlike standard approaches that assume linear state spaces, methods based on circular statistics correctly use the proper manifold and rely on probability distributions defined on this manifold. Circular statistics has been applied in a variety of sciences, such as biology~\cite{batschelet1981}, bioinformatics~\cite{mardia2008}, meteorology~\cite{fisher1987}, and geosciences~\cite{mardia1981}.

There has been some work on filtering algorithms based on circular statistics by Azmani et al.~\cite{azmani2009}, which was further investigated by Stienne et al. \cite{stienne2013}. Their work is based on the von Mises distribution and allows for recursive filtering of systems with a circular state space. However, it is limited to the identity with additive noise as the system equation and the measurement equation. The filter from \cite{azmani2009} has been applied to phase estimation of GPS signals \cite{stienne2014}, \cite{stienne2012} as well as map matching \cite{mokhtari2013}. Markovic et al. have published a similar filter \cite{markovic2014icra} based on the von Mises-Fisher distribution, a generalization of the von Mises distribution to the hypersphere.

We have previously published a recursive filter based on the wrapped normal distribution allowing for a nonlinear system equation \cite{ACC13_Kurz}. The paper \cite{ACC14_Kurz} extends this approach to make a nonlinear measurement update possible. Both papers rely on a deterministic sampling scheme, based on the first circular moment. This kind of sampling is reminiscent of the well-known unscented Kalman filter (UKF) \cite{julier2004}. We have extended this sampling scheme to the first two circular moments in \cite{Fusion14_KurzGilitschenski}, so the proposed filters are, in a sense, circular versions of the UKF. 

The developed filters have been applied in the context of constrained tracking \cite{IPIN13_Kurz}, bearings-only sensor scheduling \cite{Fusion13_Gilitschenski}, as well as circular model predictive control \cite{ACC15_Kurz}. 

Furthermore, we proposed a recursive filter based on the circular Bingham distribution in~\cite{Fusion13_Kurz-Bingham}. The Bingham distribution is closely related to the von Mises distribution, but has a bimodal density, which makes it interesting for axial estimation problems (i.e., problems with $180^\circ$ symmetry). 

An overview of all of these filters and the considered distributions as well as system and measurement models is given in Table~\ref{table:circularfilters}.

\begin{table*}
	\fontsize{8pt}{1pt}\selectfont
	\centering
	\begin{tabular}{llllll}
	\toprule
	&  & \multicolumn{2}{c}{\bf system} & \multicolumn{2}{c}{\bf measurement} \\
	\cmidrule(r){3-4}
	\cmidrule(r){5-6}
	\bf publication & \bf distribution & \bf model & \bf noise & \bf model & \bf noise \\
	\midrule
	Azmani, Reboul, Choquel, Benjelloun \cite{azmani2009} & von Mises & identity & additive & identity & additive \\
	Markovic, Chaumette, Petrovic \cite{markovic2014icra} & von Mises-Fisher & identity & additive & identity & additive \\
	Kurz, Gilitschenski, Julier, Hanebeck \cite{Fusion13_Kurz-Bingham} & Bingham & identity & additive & identity & additive \\
	Kurz, Gilitschenski, Hanebeck \cite{ACC13_Kurz} & wrapped normal/von Mises & nonlinear & additive & identity & additive \\
	Kurz, Gilitschenski, Hanebeck \cite{ACC14_Kurz} & wrapped normal & nonlinear & additive & nonlinear & any \\
	this paper & wrapped normal/von Mises & nonlinear & any & nonlinear & any \\
	\bottomrule
	\end{tabular}
	\caption{Circular filters based on directional statistics.}
	\label{table:circularfilters}
\end{table*}

This paper summarizes and combines our results as well as extends the previous work by a number of additional contributions. First of all, we propose a general filtering framework that can be used in conjunction with a variety of system and measurement equations, different types of noise, and both the wrapped normal and the von Mises distributions. Our previous publications \cite{ACC13_Kurz}, \cite{ACC14_Kurz} as well as the work by Azmani et al. \cite{azmani2009} can be seen as special cases of the proposed framework. 

Furthermore, we introduce a new multiplication formula for wrapped normal distributions that outperforms the solution proposed in \cite{ACC13_Kurz}. We generalize the prediction step from \cite{ACC13_Kurz} to a purely moment-based solution that does not need to assume any kind of distribution. Compared to \cite{ACC14_Kurz}, we add the ability to deal with non-additive noise not only in the measurement update but also in the prediction step. Finally, we perform a thorough evaluation, where we compare the proposed techniques to several state-of-the-art approaches.

This paper is structured as follows. First, we formulate the problem in Sec.~\ref{sec:problemformulation}. Then, we introduce the necessary fundamentals from circular statistics in Sec.~\ref{sec:circularstatistics}. Based on these fundamentals, we propose deterministic sampling schemes in Sec.~\ref{sec:deterministic} and derive the operations on circular densities required for the circular filter in Sec.~\ref{sec:operations}. These results are used to introduce circular filtering algorithms in Sec.~\ref{sec:circularfiltering}. An evaluation of the proposed algorithms can be found in Sec.~\ref{sec:evaluation}. Finally, conclusions are given in Sec.~\ref{sec:conclusion}.

\section{Problem Formulation}
\label{sec:problemformulation}
In this section, we formulate the problems under consideration and summarize some standard approaches that have been used to address the issues associated with periodicity.

\subsection{Circular Filtering}
\label{sec:problemformulation:circularfiltering}
Circular filtering considers estimation problems on the unit circle, which is commonly parameterized as the set of complex numbers with unit length $\{ x \in \mathbb{C} : |x| = 1 \}$. To allow for a more convenient one-dimensional notation, we identify $\Sone$ with the half-open interval $[0,2 \pi) \subset \mathbb{R}$, while keeping the topology of the circle. Together with the operation
\begin{align*}
+: \Sone \times \Sone \to \Sone, \quad x + y := x+_\mathbb{R} y \mod 2 \pi, \quad
\end{align*} 
for all  $x,y \in [0,2 \pi)$ with standard addition $+_\mathbb{R}$ on $\mathbb{R}$, the circle $\Sone$ forms an Abelian group. Because $\Sone$ with the topology given above has a manifold structure, $(\Sone, +)$ is a Lie group.

We consider a system whose state $x_k$ at time step $k$ is a value on the unit circle $\Sone$. System and measurement models are assumed to be given. In this paper, we propose several methods to deal with different types of models. More complex models necessitate the use of more sophisticated algorithms and conversely, simpler models allow the use of computationally less expensive algorithms.

\subsubsection{System Model}
In this work, we consider a system model whose state evolves according to the general system equation
\begin{align*}
x_{k+1} = a_k(x_k,w_k)
\end{align*}
with (nonlinear) system function $a_k: \Sone \times W \to \Sone$ and noise $w_k \in W$ stemming from some noise space $W$. Note that $W$ is not necessarily $\Sone$, but may be an arbitrary set, for example the real-vector space $\mathbb{R}^n$, some manifold, or even a discrete set. An interesting and practically relevant special case is a (nonlinear) system with additive noise
\begin{align*}
x_{k+1} = a_k(x_k)+ w_k \mod 2 \pi
\end{align*}
with $a_k: \Sone \to \Sone$ and $w_k \in \Sone$. More particular, we also consider the special case, where $a_k$ is the identity, i.e.,
\begin{align*}
x_{k+1} = x_k + w_k \mod 2 \pi.
\end{align*}

\subsubsection{Measurement Model}
The system state cannot be observed directly, but may only be estimated based on measurements that are disturbed by noise. A general measurement function is given by
\begin{align*}
\hat{z}_{k} = h_k(x_k,v_k) \ ,
\end{align*}
where $\hat{z}_{k} \in Z$ is the measurement in the measurement space $Z$, $h_k: \Sone \times V \to Z$ is the measurement function and $v_k \in V$ is arbitrary measurement noise in a certain noise space $V$. Note that the measurement and noise space can be arbitrary sets in general.  An interesting special case are measurement functions where the measurement noise is additive, i.e.,
\begin{align*}
\hat{z}_{k} = h_k(x_k) + v_k
\end{align*}
with measurement function $h_k: \Sone \to Z$ and $v_k \in Z$. In this case, we require $Z$ to have a group structure with $+$ as the operation. Additionally, we consider the more specific case where $h_k$ is the identity, i.e.,
\begin{align*}
\hat{z}_{k} = x_k + v_k \mod 2 \pi
\end{align*}
with $\hat{z}_{k}, v_k \in \Sone$.

\begin{remark}
We do not consider linear system models because linearity is a concept of vector spaces, not manifolds \cite{ACC13_Kurz}. For this reason, there are no linear functions on the circle.
\end{remark}

\subsection{Standard Approaches}
\label{sec:problemformulation:standardapproaches}
As circular estimation problems are wide-spread in a variety of applications, a number of standard approaches have been employed. We introduce three of the most common methods and explain their strengths and weaknesses.

\subsubsection{Gaussian-based approaches}
Gaussian-based methods (wrongly) assume a Gaussian distribution and use standard filtering techniques for Gaussians in conjunction with certain modifications to allow their application to circular problems.
\paragraph{One-Dimensional Methods}
One common approach is the use of a standard Kalman filter \cite{kalman1960} or, in case of nonlinear system or measurement functions, unscented Kalman filter (UKF) \cite{julier2004} with a scalar state $x_k$ containing the angle $\theta_k$, i.e., $x_k = \theta_k$. However, two modifications are necessary before this approach can be used in practice. First, the estimate has to be forced to stay inside the interval $[0,2 \pi)$ by performing a modulo operation after every prediction and/or update step. 

\begin{figure}
	\centering
	\begin{subfigure}{0.45\textwidth}
	\centering
	\includegraphics[width=6cm]{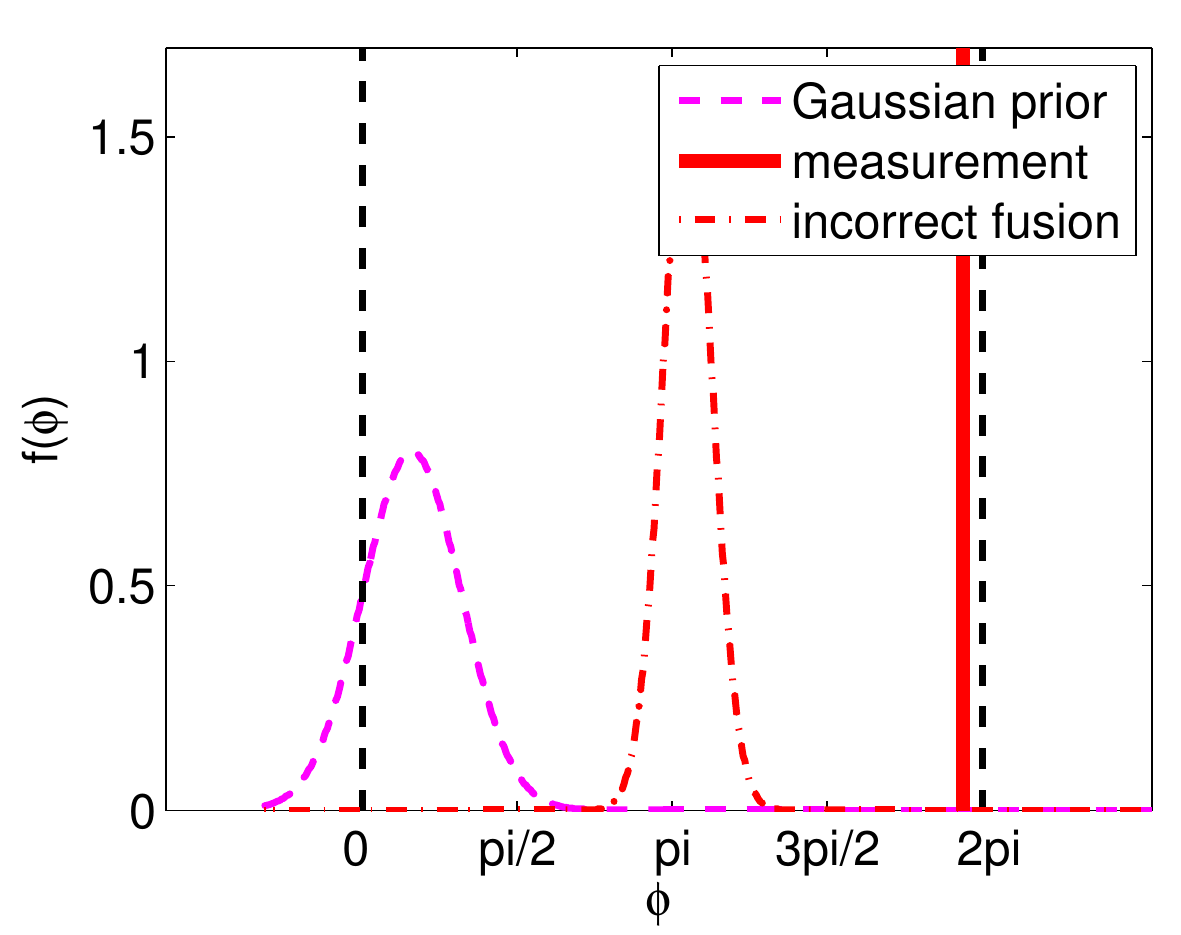}
	\caption{Incorrect fusion.}
	\end{subfigure}
	\begin{subfigure}{0.45\textwidth}
	\centering
	\includegraphics[width=6cm]{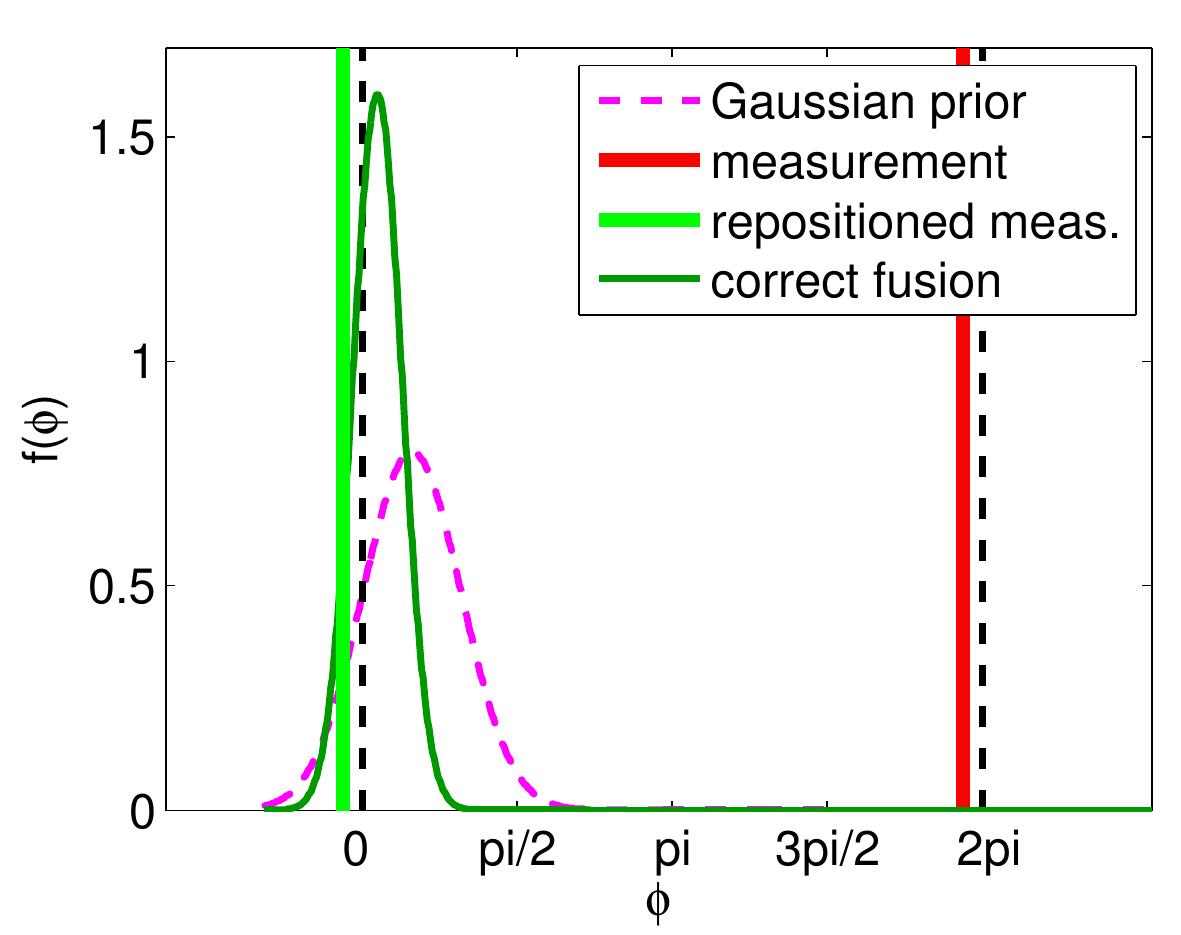}
	\caption{Correct fusion.}
	\end{subfigure}
	\caption{This figure illustrates that repositioning of the measurement is necessary to obtain satisfactory performance when using classical filters on circular problems.}
	\label{fig:measurementrepositioning}
\end{figure}

Second, if the measurement space is periodic, the measurement needs to be repositioned to be closer to the state mean in certain cases. This problem occurs whenever the measurement and the current state mean are more than $\pi$ apart. In this case, the measurement needs to be moved by $\pm 2 \pi k$ to an equivalent measurement that deviates at most $\pi$ from the state mean. An illustration of this problem is given in Fig.~\ref{fig:measurementrepositioning}.

This type of filter is used as a comparison in \cite{ACC13_Kurz}. When the uncertainty is small, this kind of approach works fairly well, but it tends to produce unsatisfactory results if the uncertainty is high. 

\paragraph{Two-Dimensional Methods}
Another common approach is based on the Kalman filter or the UKF with two-dimensional state subject to a nonlinear constraint. More specifically, an angle $\theta_k$ is represented by a state vector $\vecx_k = [\cos(\theta_k), \sin(\theta)_k]^T$ and the constraint is $||\vecx_k|| = 1$ to enforce that $\vecx_k$ is on the unit circle. In order to enforce this constraint, $\vecx_k$ is projected to the unit circle after each prediction and/or update step. More sophisticated approaches increase the covariance to reflect the fact that the projection operation constitutes an increase in uncertainty~\cite{julier2007}.

This approach has been used in \cite{IPIN13_Kurz}, but did not perform as well as the filter based on circular statistics. One of the issues of this approach is the fact that the system and measurement model sometimes become more complicated when the angle $\theta_k$ is transformed to a two-dimensional vector.

\subsubsection{Particle Filters}
Another method that can be applied is particle filtering \cite{arulampalam2002}. Particle filters on nonlinear manifolds are fairly straightforward to implement because each particle can be treated separately. For the particle filter to work, the system function and the measurement likelihood both need to respect the underlying topology. The reweighting step as well as the commonly used sequential importance resampling (SIR) are independent of the underlying manifold and can be used in a circular setting as well.

However, issues that are typically associated with particle filters arise. If the measurement likelihood function is very narrow, particle degeneration can occur, i.e., (almost) all particles have zero weight after the reweighting step. Furthermore, a large number of particles is required to obtain stable results. Even though these problems are less critical in a one-dimensional setting, there can still be issues if measurements with high certainty occur in areas with few particles. This can, for example, occur when information from sensors with very different degrees of accuracy is fused. It should also be noted that sampling from certain circular distributions can be somewhat involved and require the use of the Metropolis Hastings algorithm \cite{hastings1970} or similar approaches, e.g., in case of the von Mises distribution.

\section{Circular Statistics}
\label{sec:circularstatistics}
Because of the drawbacks of the approaches discussed above, we propose a filtering scheme based on circular statistics \cite{jammalamadaka2001}, \cite{mardia1999}. In the following, we introduce the required fundamentals from the field of circular statistics.

\subsection{Circular Distributions}
Circular statistics considers probability distributions defined on the unit circle. A variety of distributions has been proposed \cite{batschelet1981}. We give definitions of all distributions that are required for the proposed filtering scheme.

\begin{figure}
	\centering
	\begin{subfigure}{0.45\textwidth}
	\centering
	\includegraphics[width=\textwidth]{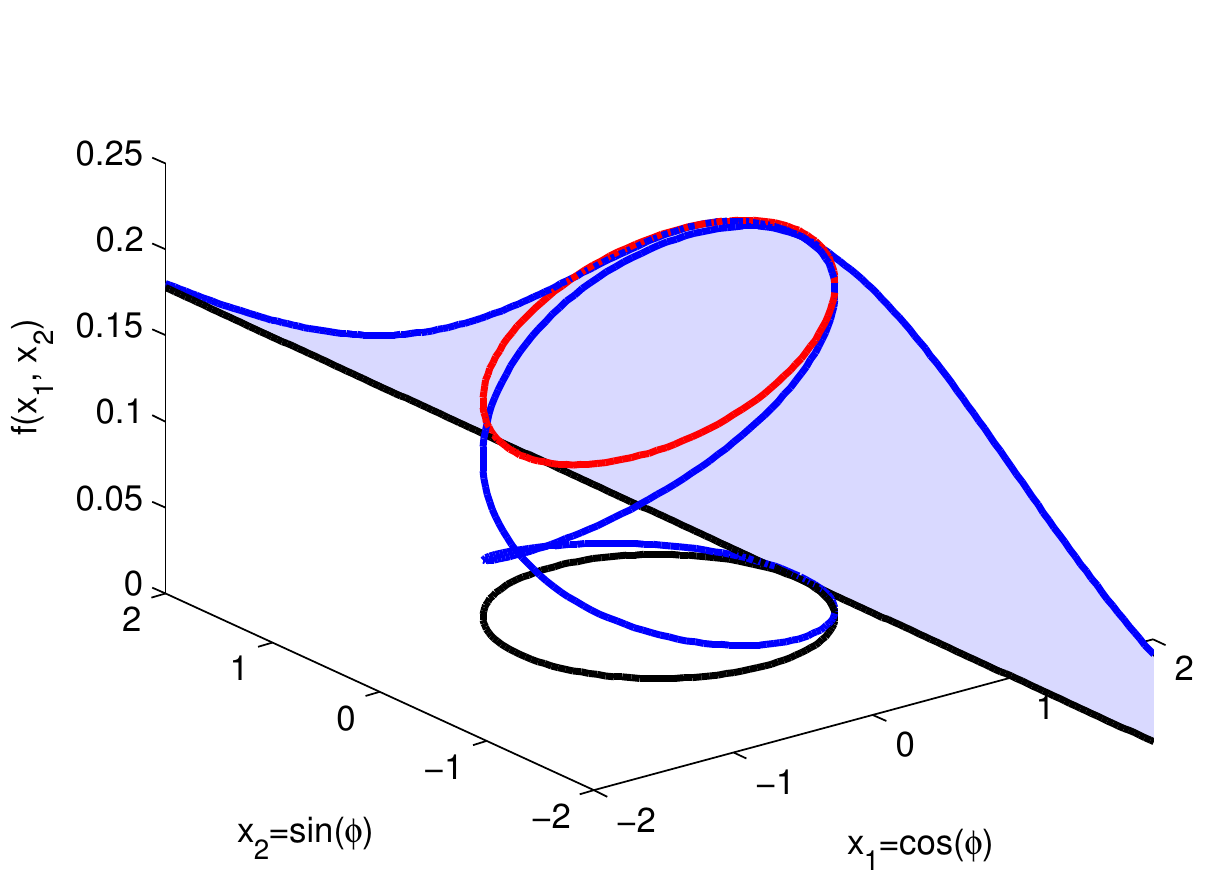}
	\caption{The wrapped normal distribution (red) is obtained by wrapping a Gaussian distribution (blue) around the circle. The parameters we used for this example are $\mu=0, \sigma=2$.}
	\label{fig:gausswn}
	\end{subfigure}
	\qquad
	\begin{subfigure}{0.45\textwidth}
	\centering
	\includegraphics[width=\textwidth]{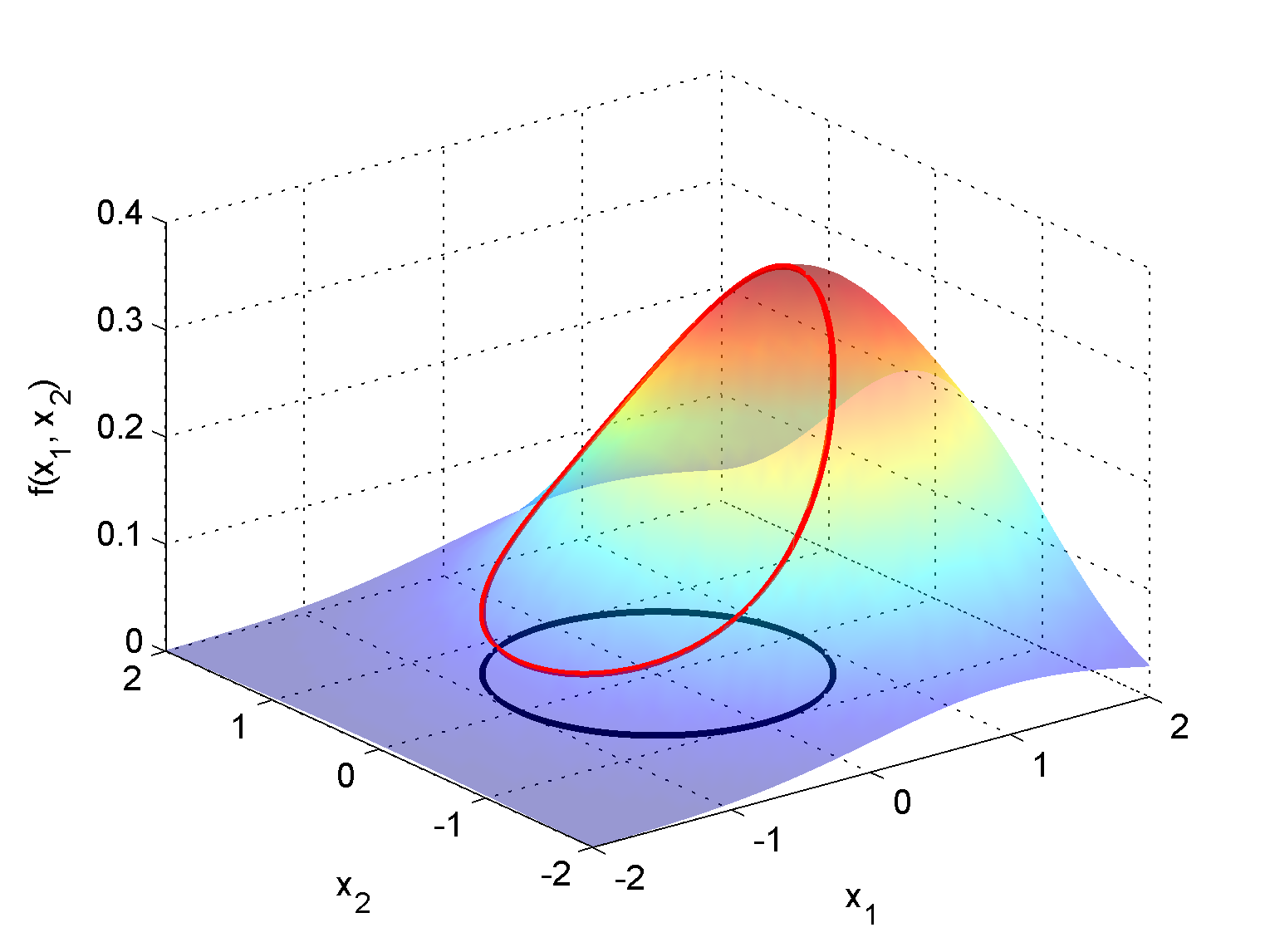}
	\caption{The von Mises distribution (red) arises when restricting a two-dimensional Gaussian with $\vecmu = (\cos \mu, \sin \mu)^T$ and covariance $\kappa \cdot \matI_{2 \times 2}$ to the unit circle. The parameters for this example are $\mu=0, \kappa=1$.}
	\label{fig:gaussvm}
	\end{subfigure}
	\caption{Relation between WN and VM distributions and the Gaussian distribution.}
\end{figure}

\begin{definition}[Wrapped Normal Distribution]
	The wrapped normal (WN) distribution is given by the probability density function (pdf)
	\begin{align*}
	f(x; \mu, \sigma) = \frac{1}{\sqrt{2\pi} \sigma} \sum_{k=-\infty}^\infty \exp \left( - \frac{(x-\mu + 2\pi k)^2}{2 \sigma^2} \right)
	\end{align*}
	with $x\in \Sone$, location parameter $\mu \in \Sone$, and concentration parameter $\sigma>0$.
\end{definition}
The WN distribution is obtained by wrapping a one-dimensional Gaussian distribution around the unit circle and adding up all probability mass that is wrapped to the same point (see Fig.~\ref{fig:gausswn}). It appears as a limit distribution on the circle \cite{ACC13_Kurz} in the following sense. A summation scheme of random variables that converges to the Gaussian distribution in the linear case, will converge to the WN distribution if taken modulo $2 \pi$. Because of its close relation to the Gaussian distribution, the WN distribution inherits a variety of its properties, for example its normalization constant as well as the formula for the convolution of densities. Even though there is an infinite sum involved, evaluation of the pdf of a WN distribution can be performed efficiently, because only few summands need to be considered \cite{SDF14_Kurz}.

\begin{definition}[Von Mises Distribution]
	The von Mises (VM) distribution is given by the pdf
	\begin{align*}
	f(x;\mu, \kappa) = \frac{1}{2 \pi I_0(\kappa)} \exp(\kappa \cos(x-\mu))
	\end{align*}
	with  $x\in \Sone$, location parameter $\mu \in \Sone$, and concentration parameter $\kappa>0$. $I_0(\cdot)$ is the modified Bessel function \cite{abramowitz1972} of order $0$.
\end{definition}
The VM distribution, sometimes also referred to as the circular normal (CN) distribution, is obtained by restricting a two-dimensional Gaussian with mean $||\vecmu||=1$ and covariance $\kappa \cdot \matI_{2 \times 2}$ (where $ \matI_{2 \times 2}$ is the identity matrix) to the unit circle and reparameterizing to $[0,2 \pi)$ as can be seen in Fig.~\ref{fig:gaussvm}. For this reason, it also inherits some of the properties of the Gaussian distribution; most importantly, it is closed under Bayesian inference. The VM distribution has been used as a foundation for a circular filter by Azmani et al. \cite{azmani2009}. It is closely related to the Bingham distribution and conversion between the two is effectively a matter of shrinking the the interval $[0,2 \pi)$ by a factor of two two, i.e., $f(2x;\mu, \kappa)$ is a Bingham distribution \cite[p. 4]{mardia1975characterizations}.

\begin{definition}[Wrapped Dirac Mixture Distribution]
	The wrapped Dirac mixture (WD) distribution with $L$ components is given by 
	\begin{align*}
	f(x;\gamma_1, \dots, \gamma_L, \beta_1, \dots, \beta_L) = \sum_{j=1}^{L} \gamma_j \delta (x-\beta_j)
	\end{align*}
	with Dirac delta distribution $\delta(\cdot )$, Dirac positions $\beta_1, \dots, \beta_L \in S^1$, and weights $\gamma_1, \dots, \gamma_L>0$ where $\sum_{j=1}^L \gamma_j = 1$.
\end{definition}
Unlike the WN and VM distributions, the WD distribution is a discrete probability distribution on a continous domain. It can be obtained by wrapping a Dirac mixture in $\mathbb{R}$ around the unit circle. WD distributions can be used as discrete approximations of continuous distributions with a finite set of samples.

In this paper, we use the following notation. We denote the density function of a WN distribution with parameters $\mu$ and $\sigma$ with $\WN(\mu, \sigma)$. If a random variable $x$ is distributed according to this WN distribution, we write $x \sim \WN(\mu, \sigma)$. The terms $\VM(\mu, \kappa)$ and $\WD(\gamma_1, \dots, \gamma_L, \beta_1, \dots, \beta_L)$ are used analogously to describe the density functions of VM and WD distributions with parameters $\mu, \kappa$ and $\gamma_1, \dots, \gamma_L, \beta_1, \dots, \beta_L$, respectively.

\subsection{Circular Moments}
In circular statistics, there is a concept called circular (or trigonometric) moment. 

\begin{definition}[Circular Moments]
	For a random variable $x \sim f(x)$ defined on the circle, its $n$-th circular moment is given by
	\begin{align*}
	m_n &= \expect{\exp(ix)^n} = \expect{\exp(inx)} \\
	&=\int_0^{2\pi} \exp(inx) \cdot f(x)\di x \in \mathbb{C}
	\end{align*}
	with imaginary unit $i$.
\end{definition}
The $n$-th moment is a complex number and, hence, has two degrees of freedom, the real and the imaginary part. For this reason, the first moment includes information about the circular mean $\arg m_1 = \atan2(\Im m_1, \Re m_1)$ as well as the concentration $|m_1| = \sqrt{(\Re m_1)^2 + (\Im m_1)^2}$ of the distribution\footnote{The term $1-|m_1|$ is sometimes referred to as circular variance.}, similar to the first two real-valued moments. It can be shown that WN and VM distributions are uniquely defined by their first circular moment \cite{jammalamadaka2001}.

\begin{remark}
Any continuous and piecewise continuously differentiable $2\pi$-periodic function $f: \mathbb{R} \to \mathbb{R}$ can be written as a Fourier series 
\begin{align*}	
f(x) &= \sum_{k=-\infty}^\infty c_k \exp(i  k  x) \\
\text{where} \quad 
c_k &= \frac{1}{2 \pi} \int_0^{2 \pi} f(x) \exp(-i k x) dx \ .
\end{align*}
If $f(\cdot)$ is the pdf of a circular probability distribution, the Fourier coefficients are closely related to the circular moments according to $c_k = \frac{1}{2 \pi} m_{-k}$.
\end{remark}

\begin{lemma}[Moments for WN, VM, and WD Distributions]
	For WN, VM, and WD distributions with given parameters, the $n$-th circular moment can be calculated according to
	\begin{align*}
	m_n^{WN} &= \exp \left( in\mu - \frac{n^2 \sigma^2}{2} \right) \ , \\
	m_n^{VM} &= \exp(in\mu) \frac{I_{|n|}(\kappa)}{I_0(\kappa)} \ , \\
	m_n^{WD} &= \sum_{j=1}^L \gamma_j \exp(in\beta_j) \ .
	\end{align*}
\end{lemma}
A proof is given in \cite{mardia1999}.

\subsection{Circular Moment Matching}
As both WN and VM distributions are uniquely defined by their first moment, it is possible to convert between them by matching the first circular moment.

\begin{lemma}[Circular Moment Matching]
	\label{lemma:circularmomentmatching}
	We define $A(x) = \frac{I_1(x)}{I_0(x)}$ as given in \cite{mardia1999}.
	\begin{enumerate}
	\item For a given first moment $m_1$, the WN distribution with this first moment has the density $\WN\left( \atan2(\Im m_1, \Re m_1), \sqrt{-2 \log \left( |m_1| \right)} \right)$.
	\item For a given first moment $m_1$, the VM distribution with this first moment has the density $\VM\left( \atan2(\Im m_1, \Re m_1),  A^{-1}(|m_1|) \right)$.
	\item For a given VM distribution with density $\VM(\mu, \kappa)$, the WN distribution with identical first moment has the density $\WN\left( \mu, \sqrt{-2 \log \left(\frac{I_1(\kappa)}{I_0(\kappa)}\right)} \right)$.
	\item For a given WN distribution with density $\WN(\mu, \sigma)$, the VM distribution with identical first moment has the density $\VM\left(\mu, A^{-1}\left(\exp \left(- \frac{\sigma^2}{2} \right)\right)\right)$
	\end{enumerate}
\end{lemma}
The proof is given in \cite{ACC13_Kurz}. Calculation of the function $A^{-1}(\cdot)$ is somewhat tricky. In \cite{ACC13_Kurz}, we use the algorithm by Amos \cite{amos1974}\footnote{Pseudocode of this algorithm is given in \cite{ACC13_Kurz}.} to calculate $A(\cdot)$ and MATLAB's \texttt{fsolve} to invert this function. Stienne et al. have proposed closed-form approximations, which can be calculated very easily but have a large approximation error \cite{stienne2012}, \cite{stienne2013}. A more detailed discussion on approximations of $A^{-1}(\cdot)$ can be found in \cite{sra2012} and \cite[Sec.~2.3]{stienne2013thesis}.

\section{Deterministic Sampling}
\label{sec:deterministic}
In order to propagate continous probability densities through nonlinear functions, it is a common technique to use discrete sample-based approximations of the continous densities. A set of samples can easily be propagated by applying the nonlinear function to each sample individually. This approach can be used for both the prediction and the measurement update step.

We distinguish between deterministic and nondeterministic sampling. Nondeterministic sampling relies on a randomized algorithm to stochastically obtain samples of a density. Typical examples include the samplers used by the particle filter \cite{arulampalam2002} or the Gaussian particle filter \cite{kotecha2003}. Deterministic sampling selects samples in a deterministic way, for example to fit certain moments (the sampler used by the UKF, \cite{julier2004}), or to optimally approximate the shape of the density (published in \cite{MFI08_Hanebeck-LCD}, this sampler is used by the  S$^2$KF, \cite{Fusion13_Steinbring}). Deterministic sampling schemes have the advantage of requiring a significantly smaller number of samples, which is why we will focus on this type of solution.

A na\"ive solution for approximating a WN density may be the application of a deterministic sampling scheme for the Gaussian distribution (such as the samplers used in \cite{julier2004}, \cite{Fusion13_Steinbring}) and subsequently wrapping the samples. Even though this technique is valid for stochastic samples, it does not provide satisfactory results for deterministic samples. In extreme cases, wrapping can cause different samples to be wrapped to the same point, grossly misrepresenting the original density. This problem is illustrated in Fig.~\ref{fig:discretization-vs-ukf}. In the case of UKF samples (Fig.~\ref{fig:ukfsamples}), it can be seen that for $\sigma \approx 2.5$, one sample is placed at $\mu$ and two samples are placed on the opposite side of the circle, i.e., the mode of the approximation is opposite to the true mode. Furthermore, for $\sigma \approx 5$ all three UKF samples are wrapped to the same position, i.e., the sample-based approximation degenerates to a distribution with single Dirac component even though the true distribution is nearly uniform. Similar issues arise in the case of S$^2$KF samples (see Fig.~\ref{fig:s2kfsamples}).

\subsection{Analytic Solutions}
\label{sec:deterministic:analytic}
First of all, we consider analytic solutions to obtain deterministic samples. These solutions are based on circular moment-matching and only provide a small, fixed number of Dirac delta components, but are extremely fast to calculate, making them a good choice for real-time applications.

In \cite{ACC13_Kurz}, we presented a method to obtain a WD approximation with three equally weighted components, which is based on matching the first circular moment (see Algorithm~\ref{algo:3diracs}). We further extended this scheme to obtain a WD with five components by matching the first as well as second circular moment (see Algorithm~\ref{algo:5diracs}), which, as we proved in \cite{Fusion14_KurzGilitschenski}, necessitates the use of different weights. Both approaches can approximate arbitrary symmetric circular densities with given circular moments.

\begin{algorithm}
	\KwIn{first circular moment $m_1$}
	\KwOut{$\WD(\gamma_1, \dots, \gamma_3, \beta_1, \dots, \beta_3)$}
	\tcc{extract $\mu$}
	$\mu \gets \atan2 (\Im m_1, \Re m_1)$\;
	\tcc{obtain Dirac positions}
	$\alpha \gets \arccos\left( \frac{3}{2} |m_1|- \frac{1}{2}\right)$\; 
	$\beta_1 \gets \mu - \alpha \mod 2 \pi $\;
	$\beta_2 \gets \mu \mod 2 \pi$\;
	$\beta_3 \gets \mu + \alpha \mod 2 \pi $\;
	\tcc{obtain weights}
	$\gamma_1, \gamma_2, \gamma_3 \gets \frac{1}{3}$\;
	\caption{Deterministic approximation with $L=3$ components.}
	\label{algo:3diracs}
\end{algorithm}

\begin{algorithm}
	\KwIn{first circular moment $m_1$, second circular moment $m_2$,\\ parameter $\lambda \in [0,1]$ with default $\lambda=0.5$}
	\KwOut{$\WD(\gamma_1, \dots, \gamma_5, \beta_1, \dots, \beta_5)$}
	\tcc{extract $\mu$}
	$\mu \gets \atan2 (\Im m_1, \Re m_1)$\;
	$m_1 \gets |m_1|$\; 
	$m_2 \gets |m_2|$\;
	\tcc{obtain weights}
	$\gamma_5^\text{min} \gets (4 m_1^2 - 4 m_1 - m2 + 1)/(4 m_1 - m_2 - 3)$\;
	$\gamma_5^\text{max} \gets (2 m_1^2 - m_2-1)/(4 m_1 - m_2 - 3)$\;
	$\gamma_5 \gets \gamma_5^\text{min} + \lambda (\gamma_5^\text{max}-\gamma_5^\text{min})$\;
	$\gamma_1, \gamma_2, \gamma_3, \gamma_4 \gets (1 - \gamma_5)/4$\;
	\tcc{obtain Dirac positions}
	$c_1 \gets \frac{2}{1-\gamma_5} (m_1-\gamma_5)$\;
	$c_2 \gets \frac{1}{1-\gamma_5} (m_2-\gamma_5)+1$\;
	$x_2 \gets (2c_1 + \sqrt{4c_1^2 - 8 (c_1^2-c_2)})/4$\;
	$x_1 \gets c_1-x_2$\;
	$\phi_1 \gets \arccos(x_1)$\;
	$\phi_2 \gets \arccos(x_2)$\;
	$(\beta_1, \dots, \beta_5) \gets \mu + (-\phi_1, +\phi_1, -\phi_2, +\phi_2, 0) \mod 2 \pi$\;
	\caption{Deterministic approximation with $L=5$ components.}
	\label{algo:5diracs}
\end{algorithm}

\begin{figure}
	\centering
	\begin{subfigure}{0.45\textwidth}
	\includegraphics[width=\textwidth]{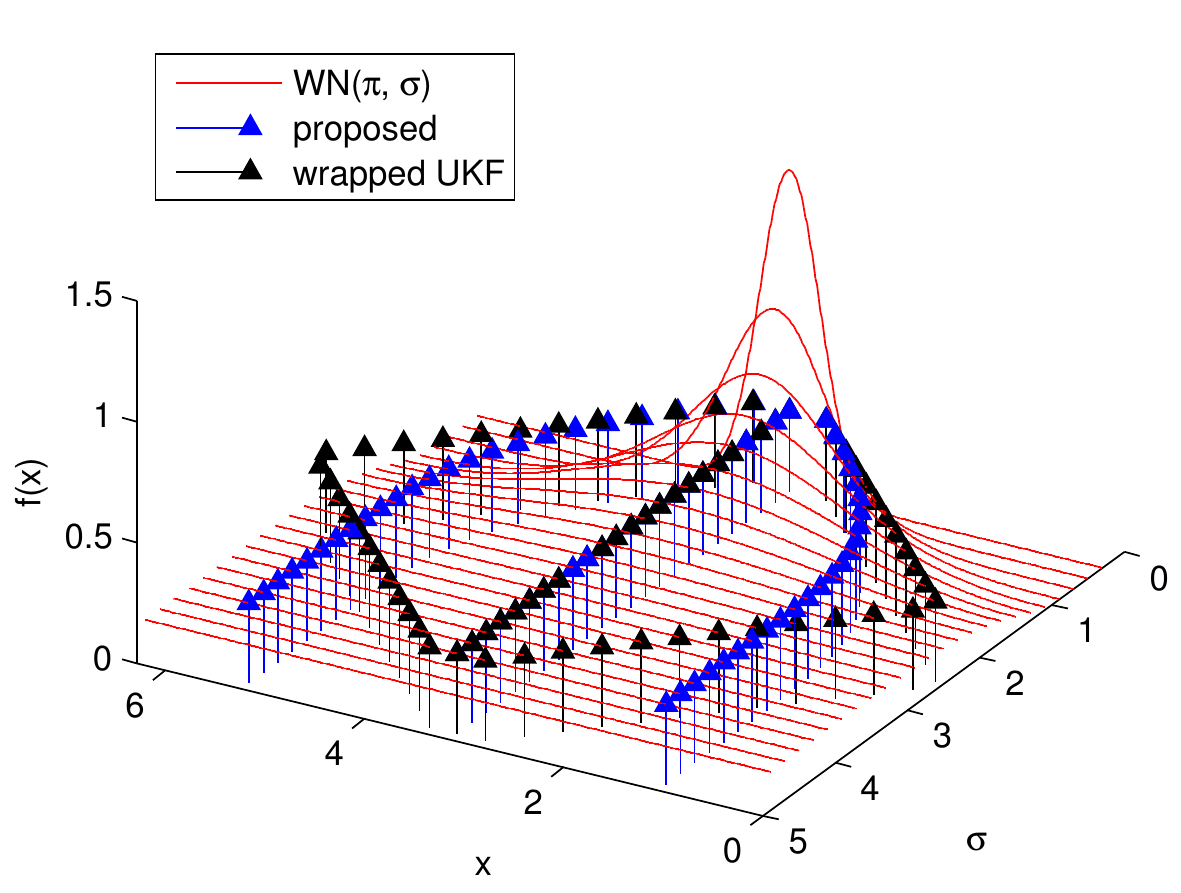}
	\caption{Analytic solution from Algorithm~\ref{algo:3diracs} and wrapped UKF \cite{julier2004} samples.}
	\label{fig:ukfsamples}
	\end{subfigure}
	\qquad
	\begin{subfigure}{0.45\textwidth}
	\includegraphics[width=\textwidth]{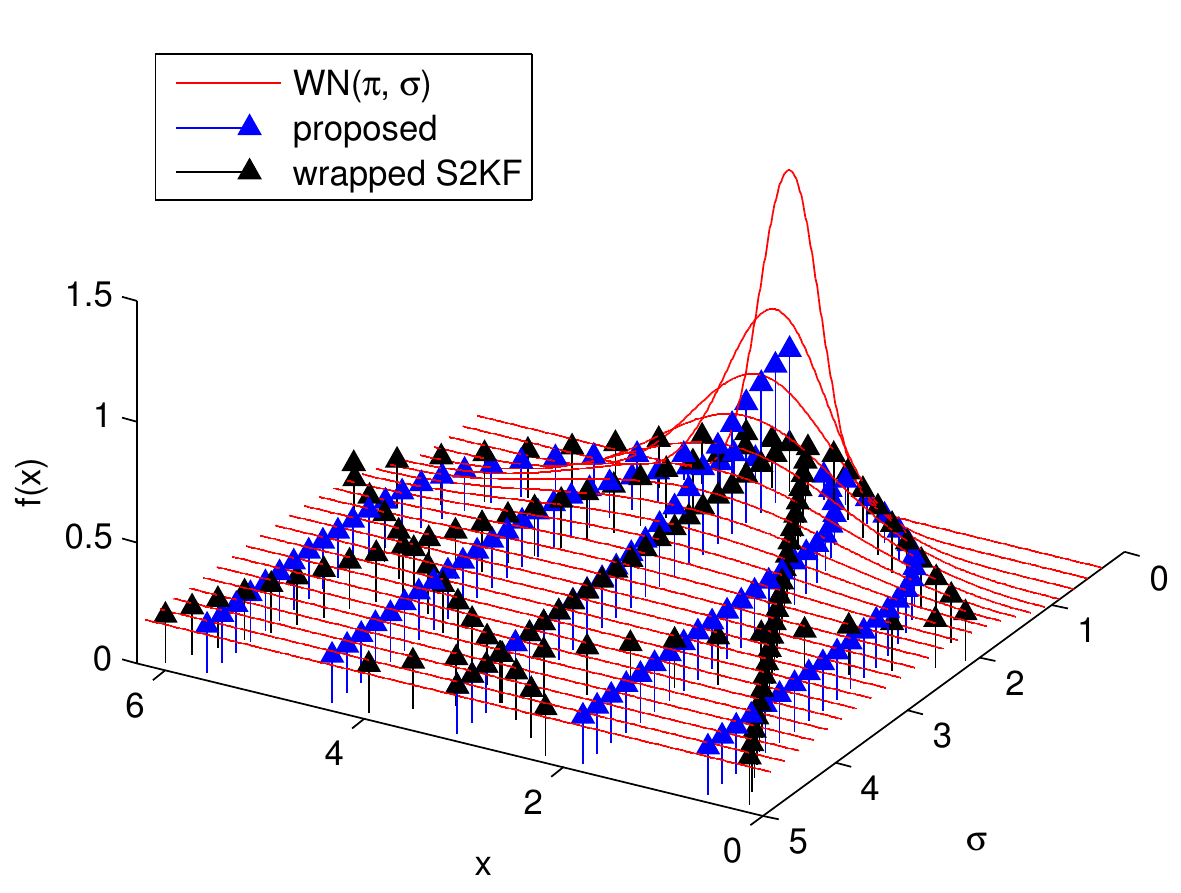}
	\caption{Analytic solution from Algorithm~\ref{algo:5diracs} with $\lambda = 0.8$ and wrapped S$^2$KF \cite{Fusion13_Steinbring} samples.}
	\label{fig:s2kfsamples}
	\end{subfigure}
	\caption{Proposed approaches for generating samples of WN distributions with a different concentration parameters $\sigma$ compared to the na\"ive approach of wrapping samples of a Gaussian with identical $\sigma$. It can be seen that the UKF and S$^2$KF samples are eventually wrapped to the same location, which produces an extremely poor approximation.}
	\label{fig:discretization-vs-ukf}
\end{figure}

\subsection{Optimization-based Solutions}
\label{sec:deterministic:optimization}
If a larger number of samples is desired and there are more degrees of freedom in the samples than constraints (such as preservation of circular moments), optimization-based solutions can be used. The number of samples can be adjusted by the user and an optimal approximation is derived by minimizing a distance measure. 

In order to simultaneously calculate optimal locations and weights for the samples, a systematic approach based on VM kernels has been proposed in \cite{IFAC14_Hanebeck}. For a WD mixture, an induced VM mixture is compared to the true distribution with a quadratic integral distance. A specific kernel width is considered for each component, which depends on the weight of the component and the value of the true distribution at the location of the component. Both the weights and the locations of a fixed even number of WD components are optimized to obtain an optimal symmetric approximation. Constraints in the optimization algorithms are used to maintain a predefined number of circular moments. This approach results in well-distributed Dirac mixtures that fulfill the moment constraints.

\begin{figure}
\centering
\includegraphics[height=65mm]{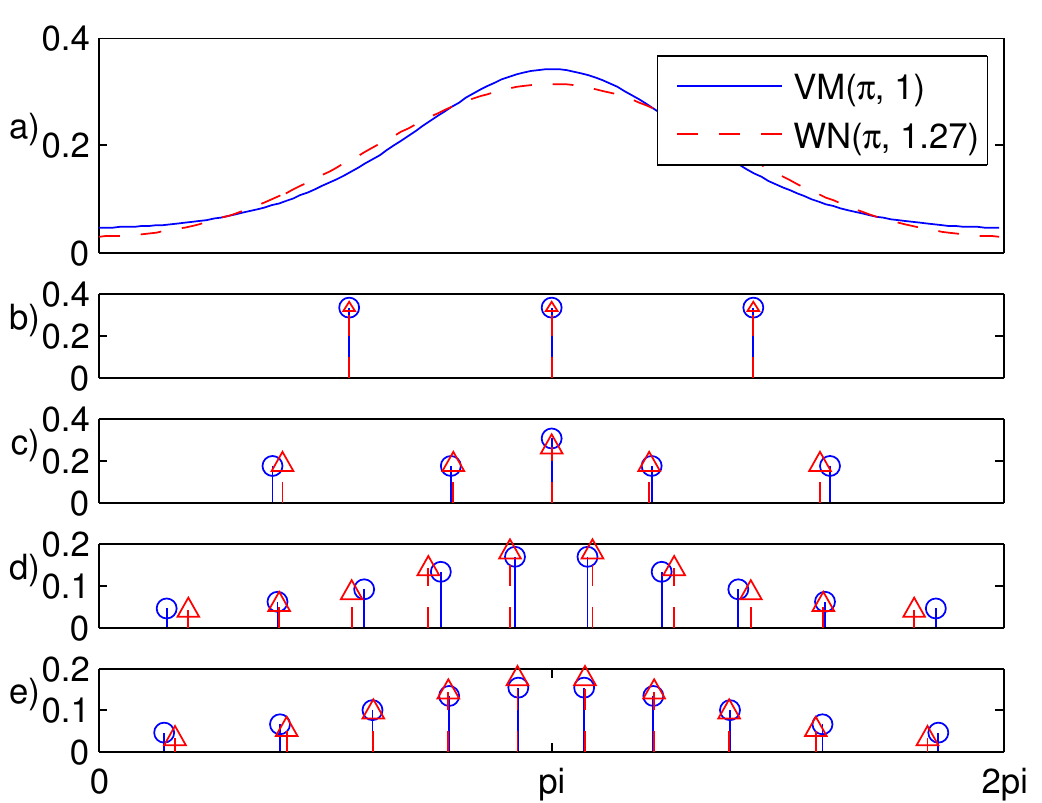}
\caption{Example of the deterministic sampling of a VM and WN distribution with equal first circular moment. From top to bottom: a) original densities, b) result of Algorithm~\ref{algo:3diracs}, c) result of Algorithm~\ref{algo:5diracs}, d) approach based on VM kernels from \cite{IFAC14_Hanebeck} for 10 components, e) quantization approach from \cite{unpublished_Quantization} for 10 components. Note that the result of Algorithm~\ref{algo:3diracs} is identical for both densities because only the first circular moment is considered. }
\label{fig:discretization}
\end{figure}

A quantization approach is discussed in \cite{unpublished_Quantization}. It is based on computing optimal Voronoi quantizers. In this approach, optimality refers to minimal quadratic distortion. The resulting Voronoi quantizer gives rise to a circular discrete probability distribution on a continuous domain that approximates the original continuous distribution. Use of this approximation is particularly beneficial in the prediction step of stochastic filters, because an error bound for propagation through a non-trivial system function can be obtained without actually knowing that function. It is sufficient to require it to be Lipschitz and to know an upper bound for the Lipschitz constant. Furthermore, circular moment constraints can be introduced in the optimization procedure of the quantization approach.

Examples from all discussed methods for deterministic sampling are depicted in Fig.~\ref{fig:discretization}.

\section{Operations on Densities}
\label{sec:operations}
In order to derive a circular filtering algorithm, we need to be able to perform certain operations on the involved probability densities.

\subsection{Shifting and Mirroring}
\label{sec:operations:shifting}
For a given density $f(x)$, we want to obtain the density $f(c-x)$ for a constant $c \in \Sone$. This operation is necessary in certain cases of the update step. We can split this operation into two steps: mirroring to obtain $f(-x)$, and subsequent shifting by $c$ to obtain $f(c + (-x))$. Mirroring $\WN(\mu, \sigma)$ and $\VM(\mu, \kappa)$ yields
$\WN(2\pi - \mu, \sigma)$ and $\VM(2 \pi - \mu, \kappa)$ because the distributions are symmetric around their mean. Shifting $\WN(\mu, \sigma)$ and $\VM(\mu, \kappa)$ by $c$ yields
\begin{align*}
\WN(\mu - c \mod 2 \pi, \sigma)\  \text{ and }\  \VM(\mu - c \mod 2 \pi, \kappa) \ ,
\end{align*}
so the combined operation results in
\begin{align*} \WN&((2 \pi- \mu) - c  \mod 2 \pi, \sigma) \\
 \text{ and }\quad  \VM&((2 \pi- \mu) - c  \mod 2 \pi, \kappa) \ .
\end{align*}

\subsection{Circular Convolution}
\label{sec:operations:convolution}
Given two independent circular random variables $x_1 \sim f_1(x_1), x_2 \sim f_2(x_2)$, the sum $x_1 + x_2$ is distributed according to 
\begin{align*}
(f_1 * f_2)(x) = \int_0^{2\pi} f_1(t) f_2(x-t) \di t \ ,
\end{align*}
where $*$ denotes the convolution. This operation is necessary in the prediction step to incorporate additive noise.

WN distributions are closed under convolution and the new pdf can be obtained just as in the Gaussian case \cite{ACC13_Kurz}, i.e., $\mu = \mu_1 + \mu_2 \mod 2 \pi, \sigma^2 = \sigma_1^2+\sigma_2^2$. VM distributions are not closed under convolution. For this reason, Azmani et al. \cite{azmani2009} used the approximation from \cite{mardia1999}, which is given by $\mu = \mu_1 + \mu_2, \kappa = A^{-1}(A(\kappa_1), A(\kappa_2))$. The function $A(\cdot)$ is the same as defined in Lemma~\ref{lemma:circularmomentmatching}. This approximation can be derived from an intermediate WN representation \cite{stienne2013}. A similar approximation has been used by Markovic et al. for the von Mises-Fisher case \cite[(7)]{markovic2014icra}.

In this paper, we present a more general result that calculates the convolution based on circular moments. 
\begin{lemma}[Moments After Addition of Random Variables]
	\label{lemma:momentsofconvolution}
	Assume independent random variables $x_1 \sim f_1, x_2 \sim f_2$ defined on the circle. For the sum $x = x_1 + x_2$, it holds
	\begin{align*}
	\expect{ \exp(inx) } = \expect{ \exp(inx_1)}  \expect{ \exp(inx_2) } \ .
	\end{align*}
\end{lemma}
\begin{proof}
	\begin{align*}
	m_n =& \expect{ \exp(inx) } = \int_0^{2\pi} \exp(inx) f(x) \di x \\
	=&\int_0^{2\pi} \int_0^{2\pi} \exp(in(x)) f_1(y) f_2(x-y) \di y  \di x \\
	=&\int_0^{2\pi} \int_0^{2\pi} \exp(in(x_1+x_2)) f_1(x_1) f_2(x_2) \di x_1  \di x_2 \\
	=&\int_0^{2\pi} \exp(inx_1) f_1(x_1) \di x_1 \int_0^{2\pi} \exp(inx_2) f_2(x_2) \di x_2 \\
	=& \expect{ \exp(inx_1)}  \expect{ \exp(inx_2) } 
	\end{align*}
\end{proof}

If moment matching of the first circular moment is used to fit a WN or VM to the density that results from convolution, the solutions for WN and VM distributions from \cite{azmani2009} and \cite{ACC13_Kurz} arise as special cases of Lemma~\ref{lemma:momentsofconvolution}.

\begin{remark}
In fact, Lemma~\ref{lemma:momentsofconvolution} allows us to calculate the convolution purely moment-based. For this reason, we do not need to assume any particular distribution, but can  just calculate the moments of the convolved density based on the products of original moments.\footnote{In general, for example in the case of mixture densities, the complexity of the density increases with each successive convolution, so considering only a finite number of moments constitutes an approximation.}
\end{remark}

\subsection{Multiplication}
\label{sec:operations:multiplication}
Multiplication of pdfs is an important operation for filtering algorithms, because it is required for Bayesian inference. In general, the product of two pdfs is not normalized and thus, not a pdf. For this reason, we consider the renormalized product, which is a valid pdf.

\subsubsection{VM}
\label{sec:operations:multiplication:vm}
Von Mises densities are closed under multiplication \cite{azmani2009}. It holds that $\VM(\mu_1, \kappa_1) \cdot \VM(\mu_2, \kappa_2) \propto \VM(\mu, \kappa)$, where
\begin{align*}
\mu = \atan2(\Im m_1, \Re m_1), \qquad \kappa = |m_1|, \\
\text{with } m_1 = \kappa_1 \exp(i\mu_1) + \kappa_2 \exp(i\mu_2) \ .
\end{align*}

\begin{figure}
\centering
\includegraphics[width=6cm]{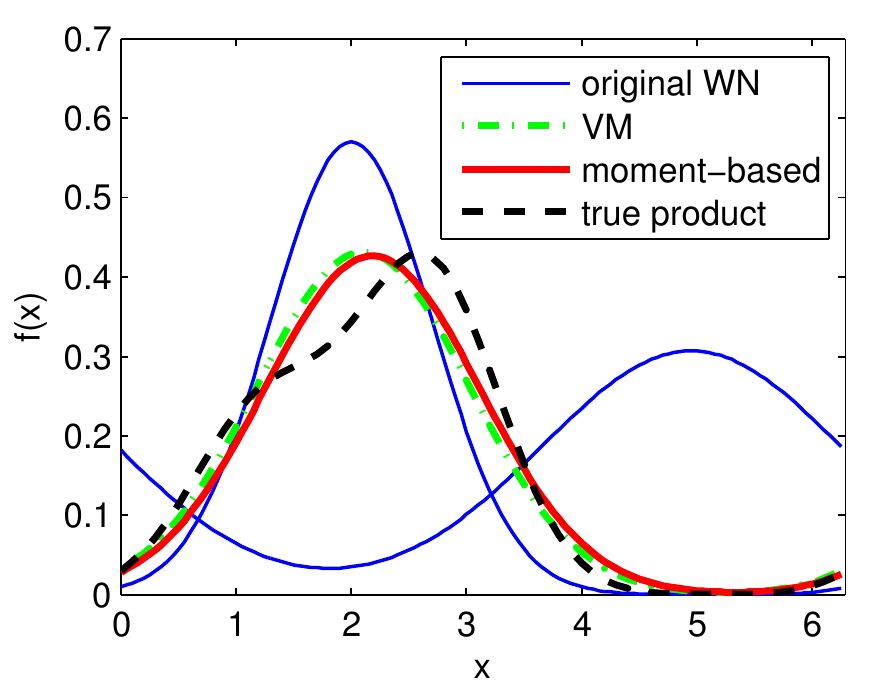}
\caption{Multiplication of two WN densities with parameters $\mu_1 = 2, \sigma_1 = 0.7,$ and \mbox{$\mu_2 = 4.95, \sigma_2 = 1.3$}. The true product and the results of both proposed approximation methods (VM and moment-based) are depicted. Note that the true product is not a WN density.}
\label{fig:wnmulexample}
\end{figure}

\subsubsection{WN}
WN densities are not closed under multiplication. In the following, we consider two different methods to approximate the density of the product with a WN density.
\paragraph{WN via VM}
\label{sec:operations:multiplication:wnvm}
WN densities are not closed under multiplication. In \cite{ACC13_Kurz}, we proposed a method to use the VM distribution in order to approximate the product of two WN densities. More specifically, we convert the WN densities to VM densities using Lemma~\ref{lemma:circularmomentmatching}, multiply according to the VM multiplication formula, and convert back to a WN distribution by applying Lemma~\ref{lemma:circularmomentmatching} again. This method has the disadvantage that, in general, the first circular moment of the resulting WN does not match the first circular moment of the true product. An example can be seen in Fig.~\ref{fig:wnmulexample}.

\paragraph{WN via Moment Matching}
\label{sec:operations:multiplication:wnmoment}
In this paper, we present a new method for approximating the product of WN distributions. This method is based on directly approximating the true posterior moments.

\begin{theorem}
The first circular moment of $\WN (\mu_1, \sigma_1) \cdot \WN(\mu_2, \sigma_2)$ after renormalization is given by
\begin{align*}
m_1 = \frac{
\sum\limits_{j,k=-\infty}^\infty  w(j,k) \int\limits_0^{2 \pi} \exp (ix) \mathcal{N}(x;\mu(j,k),\sigma(j,k)) dx
}{
\sum\limits_{j,k=-\infty}^\infty  w(j,k) \int\limits_0^{2\pi} \mathcal{N}(x;\mu(j,k),\sigma(j,k)) dx
}
\end{align*}
where $\mathcal{N}(x; \mu, \sigma)$ is a one-dimensional Gaussian density with mean $\mu$ and standard deviation $\sigma$, and
\begin{align*}
\mu(j,k) =& \frac{(\mu_1 + 2 \pi j)\sigma_2^2 + (\mu_2 + 2 \pi k)\sigma_1^2}{\sigma_1^2 + \sigma_2^2} \ , \\
\sigma(j,k) =& \sqrt{\frac{\sigma_1^2\sigma_2^2}{\sigma_1^2 + \sigma_2^2}} \ , \\ 
w(j,k) =& \frac{\exp \left( - \frac{1}{2} \frac{((\mu_1 + 2 \pi j) - (\mu_2 + 2 \pi k))^2}{\sigma_1^2 + \sigma_2^2} \right)}{\sqrt{2 \pi (\sigma_1^2 + \sigma_2^2)}}  \ . 
\end{align*}
\end{theorem}

\begin{proof}
The true renormalized product is given by $f(x) = c \cdot f(x;\mu_1,\sigma_1) \cdot f(x;\mu_2,\sigma_2)$, where $c$ renormalizes the product, i.e., 
\[ c = \left( \int_0^{2\pi} f(x;\mu_1,\sigma_1) \cdot f(x;\mu_2,\sigma_2) \di x \right)^{-1} \]

We calculate
\begin{align*}
m_1 =& c \cdot \int_0^{2 \pi} \exp (ix) \cdot f(x;\mu_1,\sigma_1) \cdot f(x;\mu_2,\sigma_2) \di x \\
=& c \cdot \int_0^{2 \pi} \exp (ix) \cdot \sum_{j=-\infty}^\infty \mathcal{N}(x;\mu_1+2\pi j,\sigma_1) \\
&\quad \cdot \sum_{k=-\infty}^\infty \mathcal{N}(x;\mu_2+2\pi k,\sigma_2) \di x \\
=& c \cdot \sum_{j=-\infty}^\infty \sum_{k=-\infty}^\infty \int_0^{2 \pi} \exp (ix) \cdot \mathcal{N}(x;\mu_1+2\pi j,\sigma_1) \\
&\quad  \cdot \mathcal{N}(x;\mu_2+2\pi k,\sigma_2) \di x \\
=& c \cdot \sum_{j=-\infty}^\infty \sum_{k=-\infty}^\infty \int_0^{2 \pi} \exp (ix) \cdot w(j,k) \\
&\quad \cdot \cdot \mathcal{N}(x;\mu(j,k),\sigma(j,k)) \di x \\
=& c \cdot \sum_{j=-\infty}^\infty \sum_{k=-\infty}^\infty \cdot w(j,k) \cdot \int_0^{2 \pi} \exp (ix) \\
&\quad \cdot \mathcal{N}(x;\mu(j,k),\sigma(j,k)) \di x \ ,
\end{align*}
where we use the dominated convergence theorem to interchange summation and integration. We use the abbreviations, 
\begin{align*}
\mu(j,k) =& \frac{(\mu_1 + 2 \pi j)\sigma_2^2 + (\mu_2 + 2 \pi k)\sigma_1^2}{\sigma_1^2 + \sigma_2^2} \ , \\
\sigma(j,k) =& \sqrt{\frac{\sigma_1^2\sigma_2^2}{\sigma_1^2 + \sigma_2^2}} \ , \\ 
w(j,k) =& \frac{\exp \left( - \frac{1}{2} \frac{((\mu_1 + 2 \pi j) - (\mu_2 + 2 \pi k))^2}{\sigma_1^2 + \sigma_2^2} \right) }{\sqrt{2 \pi (\sigma_1^2 + \sigma_2^2)}}
\end{align*}
based on the multiplication formula for Gaussian densities given in \cite[8.1.8]{petersen2012}.

To obtain the renormalization factor $c^{-1}$, we use a similar derivation
\begin{align*}
c^{-1} =& \int_0^{2\pi} f(x;\mu_1,\sigma)_1 \cdot f(x;\mu_2,\sigma_2) \di x \\
=& \int_0^{2\pi} \sum_{j=-\infty}^\infty \mathcal{N}(x;\mu_1+2\pi j,\sigma_1) \\
&\quad \cdot \sum_{k=-\infty}^\infty \mathcal{N}(x;\mu_2+2\pi k,\sigma_2) \di x \\
=& \sum_{j=-\infty}^\infty \sum_{k=-\infty}^\infty \int_0^{2\pi} \mathcal{N}(x;\mu_1+2\pi j,\sigma_1) \\
&\quad \cdot \mathcal{N}(x;\mu_2+2\pi k,\sigma_2) \di x \\
=& \sum_{j=-\infty}^\infty \sum_{k=-\infty}^\infty  w(j,k) \\
&\quad \cdot \int_0^{2\pi} \mathcal{N}(x;\mu(j,k),\sigma(j,k)) \di x  \ . 
\end{align*}

\end{proof}

The involved integrals can be reduced to evaluations of the complex error function erf \cite[7.1]{abramowitz1972}. This yields
\begin{align*}
&\int_0^{2 \pi} \exp (ix) \cdot \mathcal{N}(x;\mu(j,k),\sigma(j,k)) \di x \\
=& \frac{1}{2} \exp \left( i \mu(j,k) - \frac{\sigma(j,k)^2}{2} \right) \\
&\cdot \left( \erf \left( \frac{\mu(j,k)+i\sigma(j,k)^2}{\sqrt{2} \sigma(j,k)} \right) \right. \\
&- \left. \erf \left( \frac{\mu(j,k)-2\pi + i\sigma(j,k)^2}{\sqrt{2}\sigma(j,k)} \right)  \right)
\end{align*}
and
\begin{align*}
&\int_0^{2 \pi} \mathcal{N}(x;\mu(j,k),\sigma(j,k)) \di x \\ 
&= \frac{1}{2} \left( \erf \left(\frac{\mu(j,k) }{\sigma(j,k) \sqrt{2}} \right) - \erf \left( \frac{\mu(j,k) - 2\pi}{\sigma(j,k) \sqrt{2}} \right) \right) \ .
\end{align*}
There are efficient implementations of the complex error function by means of the related Faddeeva function \cite{johnson2012}. Furthermore, the infinite sums can be truncated to just a few summands without a significant loss in accuracy. For example, the multiplication in Fig.~\ref{fig:wnmulexample} requires $5 \times 5$ summands for an error smaller than the accuracy of the IEEE 754 64 bit double data type \cite{goldberg1991}. Consequently, the proposed method allows for efficient calculation of the approximate multiplication of WN densities.

\section{Circular Filtering}
\label{sec:circularfiltering}
Based on the results in the previous section, we derive circular filtering algorithms for the scenarios described in Sec.~\ref{sec:problemformulation:circularfiltering}. All proposed algorithms follow the recursive filtering concept and consist of prediction and measurement update steps. We formulate the necessary steps without requiring a particular density whenever possible such that most methods can be directly applied to WN as well as VM distributions, and might even be generalized to other continous circular distributions. An overview of all considered prediction and measurement update algorithms is given in Table~\ref{table:predictionupdate}.

\begin{table*}
\fontsize{7pt}{1pt}\selectfont
\centering
\begin{tabular}{lcccccc}
\toprule
& \multicolumn{3}{c}{\bf system model} &  \multicolumn{3}{c}{\bf measurement model} \\
\cmidrule(r){2-4}
\cmidrule(r){5-7}
\bf method & \bf identity & \bf additive noise & \bf arbitrary noise & \bf identity & \bf additive noise & \bf arbitrary noise \\
\midrule
WN & \cite{ACC13_Kurz} (special case) & \cite{ACC13_Kurz} & contribution & \cite{ACC13_Kurz} & \cite{ACC14_Kurz} & \cite{ACC14_Kurz} \\
VM & \cite{azmani2009} & contribution & contribution & \cite{azmani2009} & contribution & contribution \\
moment-based & contribution & contribution & contribution & - & - & - \\
\bottomrule
\end{tabular}
\caption{Prediction and measurement update algorithms. Entries marked with \emph{contribution} are contributions of this paper.}
\label{table:predictionupdate}
\end{table*}

\subsection{Prediction}
The prediction step is used to propagate the estimate through time.

\subsubsection{Identity System Model}
The transition density is given according to
\begin{align*}
f(x_{k+1}|x_k) &= \int_0^{2 \pi} f(x_{k+1},w_k|x_k) \di w_k \\
&= \int_0^{2 \pi} f(x_{k+1}|x_k,w_k) f^w(w_k) \di w_k \\
&= \int_0^{2 \pi} \delta(x_{k+1}- (x_k+ w_k)) f^w(w_k) \di w_k \\
&= f^w(x_{k+1}-x_k) \ , 
\end{align*}
where $f^w(\cdot)$ is the density of the system noise. For the predicted density, according to the Chapman-Kolmogorov equation we obtain
\begin{align}
\label{eq:chapman}
f^p(x_{k+1}) &= \int_0^{2\pi} f(x_{k+1}|x_k) f^e(x_k) \di x_k \ .
\end{align}
In the special case of an identity system model, this yields
\begin{align}
f^p(x_{k+1}) &= \int_0^{2\pi} f^w(x_{k+1}-x_k) f^e(x_k) \di x_k \nonumber \\
&= (f^w * f^e ) (x_{k+1}) \nonumber \ ,
\end{align}
where $*$ denotes convolution as defined in Sec.~\ref{sec:operations:convolution}. For the VM distribution, this system model has been considered in~\cite{azmani2009}. If a WN  distribution is assumed, (\ref{eq:chapman}) is a special case of \cite{ACC13_Kurz} where we omit the propagation through the nonlinear function.

\subsubsection{Nonlinear System Model with Additive Noise}

Similar to the previous case, the transition density is given by
\begin{align*}
f(x_{k+1}|x_k)&= \int_0^{2 \pi} \delta(x_{k+1}- (a_k(x_k) + w_k)) f^w(w_k) \di w_k \ . 
\end{align*}
We approximate the prior density $f^e(x_k) \approx \sum_{j=1}^L \gamma_j \delta(x_k - \beta_j)$ using, for example, Algorithm~\ref{algo:3diracs} or Algorithm~\ref{algo:5diracs}. Then, the prediction density can be approximated according to
\begin{align*}
&f^p(x_{k+1}) \\
&= \int_0^{2\pi} f(x_{k+1}|x_k) f^e(x_k) \di x_k \displaybreak[3] \\
&\approx \int_0^{2\pi} \left( \int_0^{2 \pi} \delta(x_{k+1}- (a_k(x_k)+ w_k)) f^w(w_k) \di dw_k \right) \\
&\quad \cdot \left( \sum_{j=1}^L \gamma_j \delta(x_k - \beta_j) \right) \di x_k \\
&= \int_0^{2\pi} f^w(w_k) \sum_{j=1}^L \gamma_j \int_0^{2 \pi} \delta(x_{k+1}- (a_k(x_k)+ w_k)) \\
&\quad \cdot \delta(x_k - \beta_j) \di x_k \di w_k \\
&= \int_0^{2\pi} f^w(w_k) \underbrace{\sum_{j=1}^L \gamma_j \delta(x_{k+1}- (a_k(\beta_j)+ w_k))}_{\approx\tilde{f}(x_{k+1}-w_k)} \di w_k \\
&= (f^w * \tilde{f})(x_{k+1}) \ ,
\end{align*}
where $\tilde{f}$ is obtained by moment matching of the WD $\sum_{j=1}^L \gamma_j \delta(x - (a_k(\beta_j))$ according to Lemma~\ref{lemma:circularmomentmatching}. The convolution can be calculated as described in Sec.~\ref{sec:operations:convolution}. 

\subsubsection{Nonlinear System Model with Arbitrary Noise}
In this paper, we extend the previous results to deal with arbitrary noise in the prediction step. 

For arbitrary noise, the transition density is given by
\begin{align*}
f(x_{k+1}|x_k)&= \int_0^{2 \pi} \delta(x_{k+1}- a_k(x_k,w_k)) f^w(w_k) \di w_k \ .
\end{align*}
We approximate the prior density $f^e(x_k) \approx \sum_{j=1}^L \gamma_j \delta(x_k - \beta_j)$ as well as the noise density $f^w(w_k) \approx \sum_{l=1}^{L^w} \gamma^w_l \delta(w_k - \beta^w_l)$. It should be noted that the noise is not necessarily a circular quantity and different approximation techniques may be required. If $W=\mathbb{R}^n$, the techniques presented in \cite{MFI08_Hanebeck-LCD} may be applied. Then, the prediction density can be approximated according to
\begin{align*}
&f^p(x_{k+1}) \\
&= \int_0^{2\pi} f(x_{k+1}|x_k) f^e(x_k) \di x_k \\
&= \int_0^{2\pi} \int_W \delta(x_{k+1}- a_k(x_k,w_k)) f^w(w_k) \di w_k f^e(x_k) \di x_k \\
&\approx \int_0^{2\pi} \int_W \delta(x_{k+1}- a_k(x_k,w_k))\left(\sum_{l=1}^{L^w} \gamma^w_l \delta(w_k - \beta^w_l)\right) \\
&\quad \di w_k \ f^e(x_k) \di x_k \\
&= \sum_{l=1}^{L^w} \gamma^w_l  \int_0^{2 \pi} \delta(x_{k+1}- a_k(x_k,\beta^w_l)) f^e(x_k) \di x_k \\
&\approx \sum_{l=1}^{L^w} \gamma^w_l  \int_0^{2 \pi} \delta(x_{k+1}- a_k(x_k,\beta^w_l)) \\
&\quad \cdot  \left(\sum_{j=1}^L \gamma_j \delta(x_k - \beta_j)\right) \di x_k \\
&= \sum_{j=1}^L \sum_{l=1}^{L^w} \gamma_j \gamma^w_l \delta(x_{k+1}- a_k(\beta_j,\beta^w_l)) 
\end{align*}
A continuous density can be fitted to this result by circular moment matching. The algorithm is given in Algorithm~\ref{algo:predcitionarbitrarynoise}. It is worth noting that this algorithm can be executed based purely on the circular moments of $f(x_k)$ and $f^w(w_k)$ if the deterministic sampling scheme only depends on these moments. In that case, we do not necessarily need to fit a distribution $f(x_{k+1})$ to the resulting circular moments, but can just store the estimate by retaining those circular moments.

\begin{algorithm}
\KwIn{prior density $f(x_k)$, system noise density $f^w(w_k)$, system function $a_k(\cdot, \cdot)$}
\KwOut{predicted density $f(x_{k+1})$}
\tcc{sample prior density and noise density}
$(\gamma_1, \dots, \gamma_L, \beta_1, \dots, \beta_L) \gets$ sampleDeterm$(f(x_k))$\;
$(\gamma_1^w, \dots, \gamma_{L^w}^w, \beta_1^w, \dots, \beta_{L^w}^w) \gets$ sampleDeterm$(f^w(w_k))$\;
\tcc{obtain Cartesian product and propagate}
\For{$j \gets 1$ \KwTo $L$}{
	\For{$l \gets 1$ \KwTo $L^w$}{
		$\gamma_{j+L(l-1)}^p \gets \gamma_j \cdot \gamma_l^w$\;
		$\beta_{j+L(l-1)}^p \gets a_k(\beta_j, \beta_l^w)$\;
	}
}
\tcc{obtain posterior density}
$f(x_{k+1}) \gets$ momentMatching$(\gamma_1^p, \dots, \gamma_{L\cdot L^w}^p, \beta_1^p, \dots, \beta_{L \cdot L^w}^p))$\;
\caption{Prediction with arbitrary noise.}
\label{algo:predcitionarbitrarynoise}
\end{algorithm}

\subsection{Measurement Update}
The measurement update step fuses the current estimate with a measurement that was obtained according to the measurement equation.

\subsubsection{Identity Measurement Model}
In the case of the identity measurement model and additive noise, the measurement likelihood is given by
\begin{align*}
f(z_k | x_k) = f^v ( z_k - x_k) \ .
\end{align*}
For the posterior density, application of Bayes' theorem yields
\begin{align*}
f(x_k | z_k ) = \frac{f(z_k | x_k) f(x_k) }{f(z_k)} \propto f(z_k | x_k) f(x_k)  \ ,
\end{align*}
where $f(x_k)$ is the prior density. Thus, we obtain the posterior density
\begin{align*}
f(x_k | z_k ) \propto f^v ( z_k - x_k) f(x_k) \ ,
\end{align*}
as the product of the prior density and $f^v ( z_k - x_k)$, which can be obtained as described in  Sec.~\ref{sec:operations:shifting}. The multiplication depends on the assumed probability density and can be performed using the multiplication formulas given in Sec.~\ref{sec:operations:multiplication}. For the VM case, this is equivalent to the measurement update from \cite{azmani2009} and for the WN case, this is equivalent to the measurement update from \cite{ACC13_Kurz}.


\subsubsection{Nonlinear Model with Additive Noise}
For a nonlinear measurement function with additive noise, the measurement likelihood is calculated according to
$f(\vecz_k | x_k) = f^v ( \vecz_k - h_k(x_k)),$
as given in \cite{ACC14_Kurz}. The remainder of the measurement update step is identical to the case of arbitrary noise as described in the following section.

\subsubsection{Nonlinear Model with Arbitrary Noise}
For a nonlinear update with arbitrary noise, we assume that the likelihood is given. The key idea is to approximate the prior density $f(x_k)$ with a WD mixture and reweigh the components according to the likelihood $f(z_k | x_k)$. However, this can lead to degenerate solutions, i.e., most or all weights are close to zero, if the likelihood function is narrow. 

We have shown in \cite{ACC14_Kurz} that a progressive solution as introduced in \cite{Fusion13_Hanebeck} can be used to avoid this issue. For this purpose, we formulate the likelihood as a product of likelihoods
\begin{align*}
f(\vecz_k | x_k) = f(\vecz_k | x_k)^{\lambda_1} \cdot \ldots \cdot f(\vecz_k | x_k)^{\lambda_s} \ ,
\end{align*}
where $\lambda_1, \dots, \lambda_s > 0$ and $\sum_{j=1}^s \lambda_j = 1$. This decomposition of the likelihood allows us to perform the measurement update step gradually by performing $s$ partial update steps. Each update step is small enough to prevent degeneration and we obtain a new sample set after each step, to ensure that the differences between the sample weights stay small. In order to determine $\lambda_1, \dots, \lambda_s$ and $s$, we require that the quotient between the smallest weight $\gamma_\text{min}$ and the largest weight $\gamma_\text{max}$ after reweighing is not below a certain threshold $R \in (0,1)$, i.e., $\frac{\gamma_\text{min}}{\gamma_\text{max}} \geq R$. Using the conservative bounds 
\begin{align*}
\gamma_\text{min} &\geq \min_{j=1,\dots, L} (\gamma_j^n) \cdot \min_{j=1, \dots, L} (f(\vecz_k| \beta_j^n)) \ , \\
\gamma_\text{max} &\leq \max_{j=1,\dots, L} (\gamma_j^n) \cdot \max_{j=1, \dots, L} (f(\vecz_k| \beta_j^n)) \ ,
\end{align*}
this leads to the condition
\begin{align*}
\lambda_n \leq \frac{\log \left( R \cdot \frac{\max\limits_{j=1,\dots, L} (\gamma_j^n)}{\min\limits_{j=1,\dots, L} (\gamma_j^n)} \right)}{\log \left( \frac{\min\limits_{j=1,\dots, L}  f(\vecz_k| \beta_j^n)}{\max\limits_{j=1,\dots, L} f(\vecz_k| \beta_j^n)} \right)} \ ,
\end{align*}
where $\WD(\gamma_1^n, \dots, \gamma_L^n, \beta_1^n, \dots, \beta_L^n)$ is the deterministic approximation at $n$-th progression step\footnote{Compared to \cite{ACC14_Kurz}, we extend the progressive scheme to handle discrete approximations with non-equally weighted components.}. The progression continues until $\sum_{n}\lambda_n = 1$. This method can be applied in conjunction with WN as well as VM distributions (see Algorithm~\ref{algo:progressive}).

\begin{algorithm}[htb]
	\KwIn{measurement $\hat{\vecz}_k$, likelihood $f(\vecz_k|\vecx_k)$, predicted density $f^p(x_k)$ as WN or VM, threshold parameter $R$}
	\KwOut{estimated fensity as WN or VM}
	$s \gets 0$\;
	$f(x_k) \gets f^p(x_k)$\;
	\While{$\sum_{n=1}^{s} \lambda_n < 1$}{
		$s \gets s + 1$\;
		\tcc{deterministic sampling from WN or VM density (Sec.~\ref{sec:deterministic})}
		$(\gamma_1, \dots, \gamma_L, \beta_1, \dots, \beta_L) \gets$ sampleDeterm$(f(x_k))$\;
		\tcc{calculate size of progression step}
		$\lambda_s \gets \min \left( 1 - \sum_{n=1}^{s-1} \lambda_n ,  \frac{\log\left( R \cdot \frac{\max\limits_{j=1,\dots, L} (\gamma_j)}{\min\limits_{j=1,\dots, L} (\gamma_j)}  \right)}{\log \left( \frac{\min\limits_{j=1,\dots, L}  f(\vecz_k| \beta_j)}{\max\limits_{j=1,\dots, L} f(\vecz_k| \beta_j)} \right) } \right)$\;
		\tcc{execute progression step}
		\For{$j \gets 1$ \KwTo $L$}{
			$\gamma_j \gets \gamma_j \cdot  f(\hat{\vecz}_k | \beta_j)^{\lambda_s}$\;
		}
		\tcc{use moment matching to obtain WN or VM density}
		$f \gets $momentMatching$(\gamma_1, \dots, \gamma_L, \beta_1, \dots, \beta_L)$\;
	}
	\caption{Progressive measurement update for arbitrary noise.}
	\label{algo:progressive}
\end{algorithm}

\section{Evaluation}
\label{sec:evaluation}
\subsection{Propagation Accuracy}

\def\movetop#1{\raisebox{-\height+\baselineskip}{#1}}
\begin{figure*}[ht]
\centering
\begin{tabular}{ccc}
\movetop{\includegraphics[width=50mm]{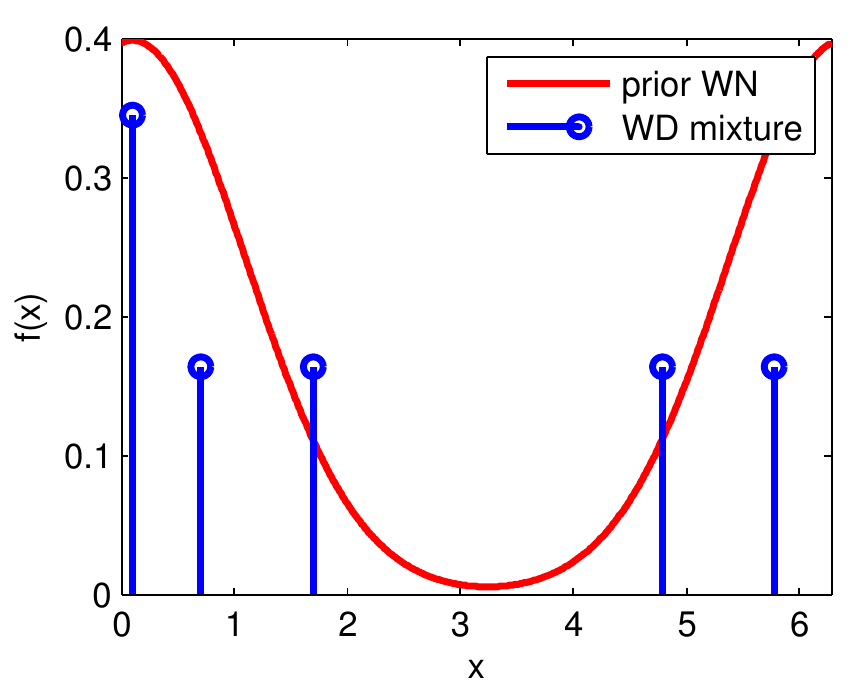}} & 
\movetop{$\substack{\includegraphics[width=3cm]{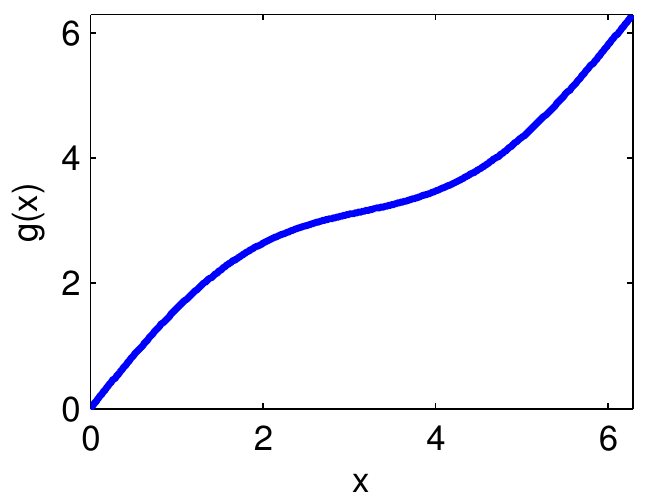} \\ \xrightarrow{\hspace*{5cm}} \\ g(x)}$} & 
\movetop{\includegraphics[width=50mm]{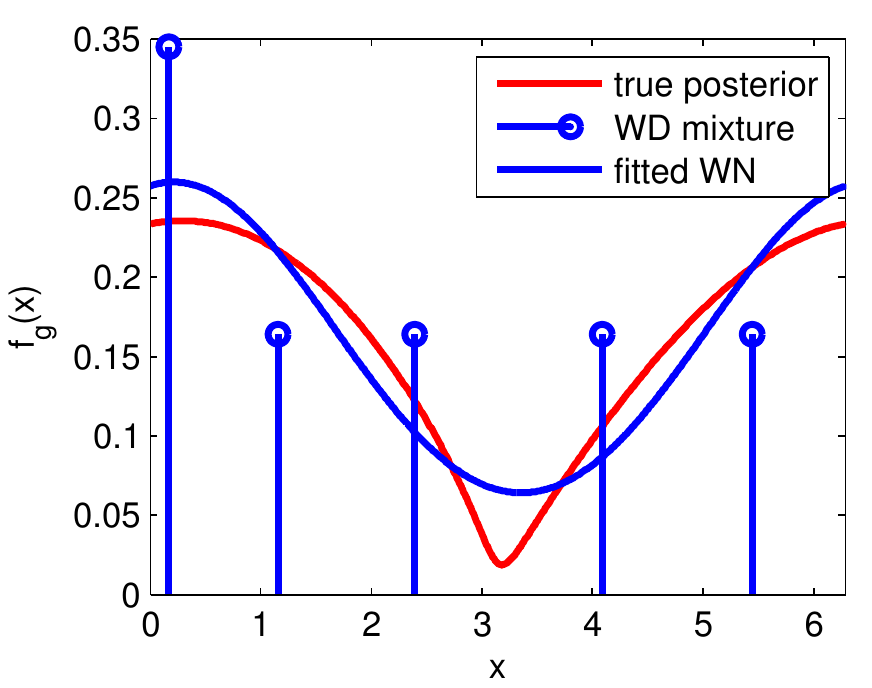}}
\\
prior & $g(x)= x + c \sin(x) \mod 2 \pi$ & posterior
\end{tabular}
\caption{Propagation of a WN distribution with parameters $\mu=0.1, \sigma=1$ through a nonlinear function $g$ by means of a deterministic WD approximation with five components. In this example, we use $c=0.7$.}
\label{fig:propagationidea}
\end{figure*}

\newcommand{\propagationwidth}{45mm}
\begin{figure*}[htb]
	\centering
	\begin{tabular}{ccc}
	$\sigma = 0.5$ & $\sigma = 1.0$ & $\sigma = 1.5$ \\
	\includegraphics[width=\propagationwidth]{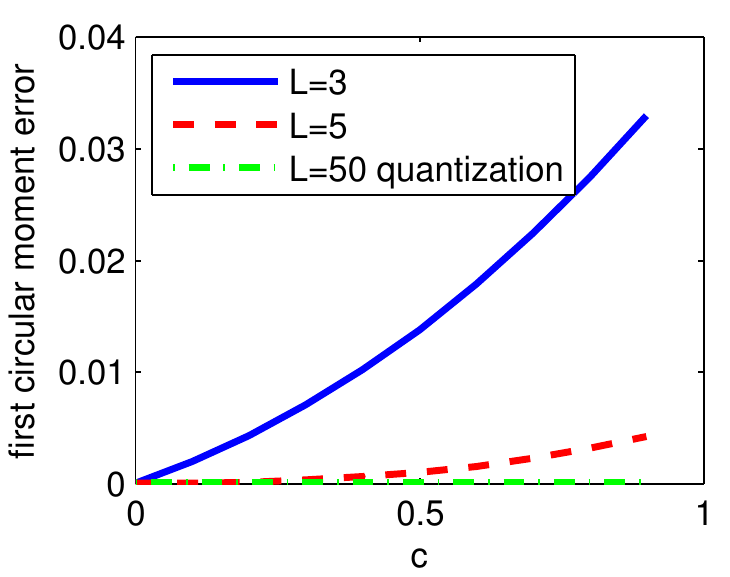} & \includegraphics[width=\propagationwidth]{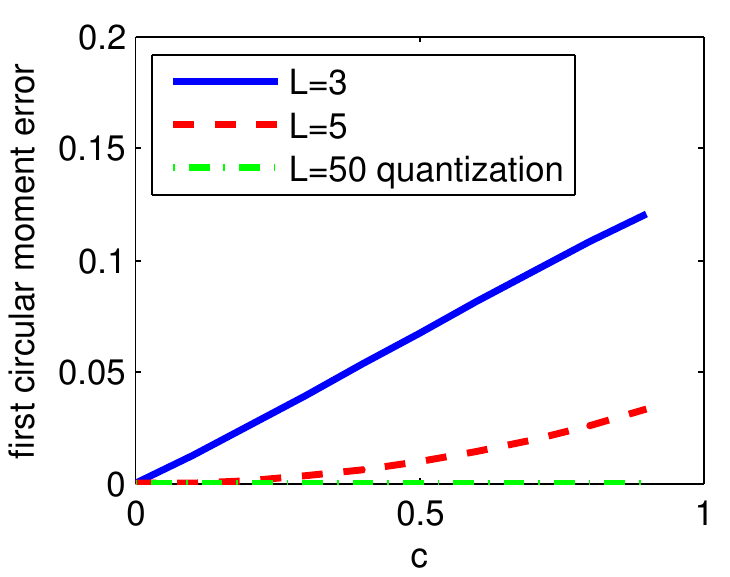} & 
	\includegraphics[width=\propagationwidth]{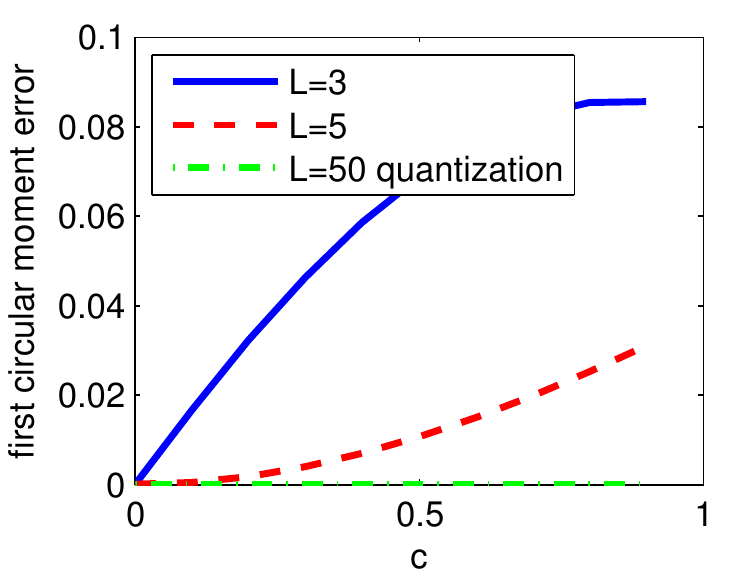} \\
	\includegraphics[width=\propagationwidth]{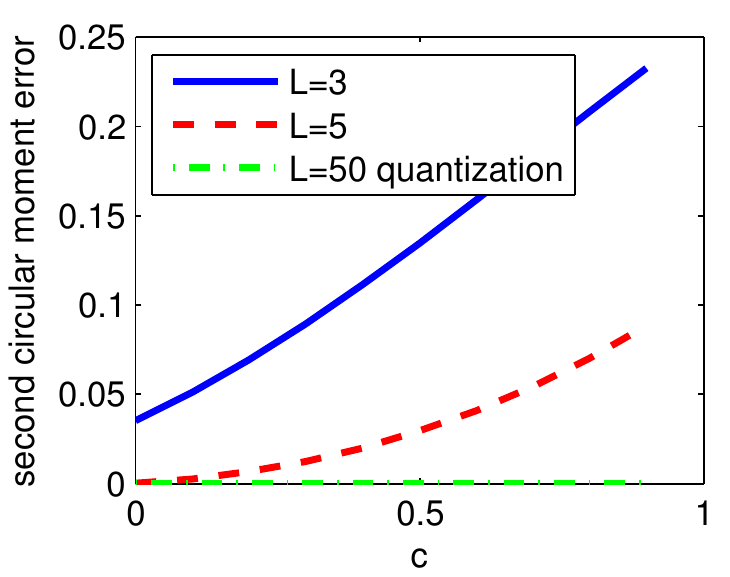} & \includegraphics[width=\propagationwidth]{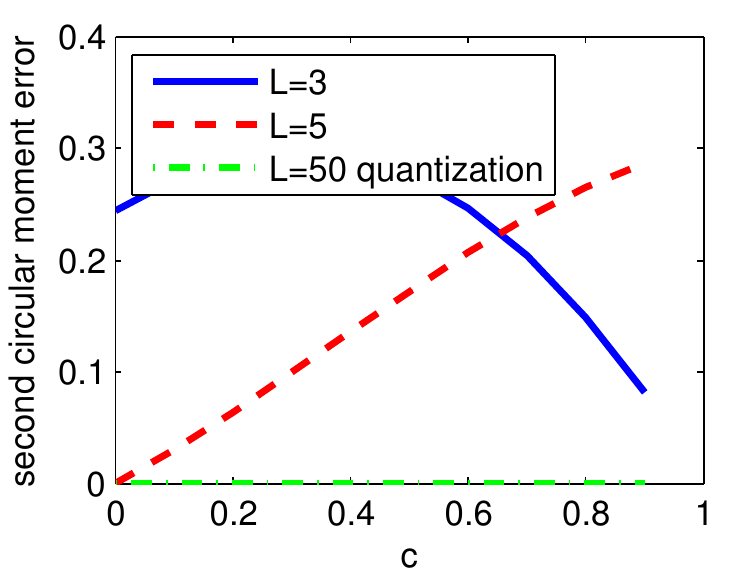} & 
	\includegraphics[width=\propagationwidth]{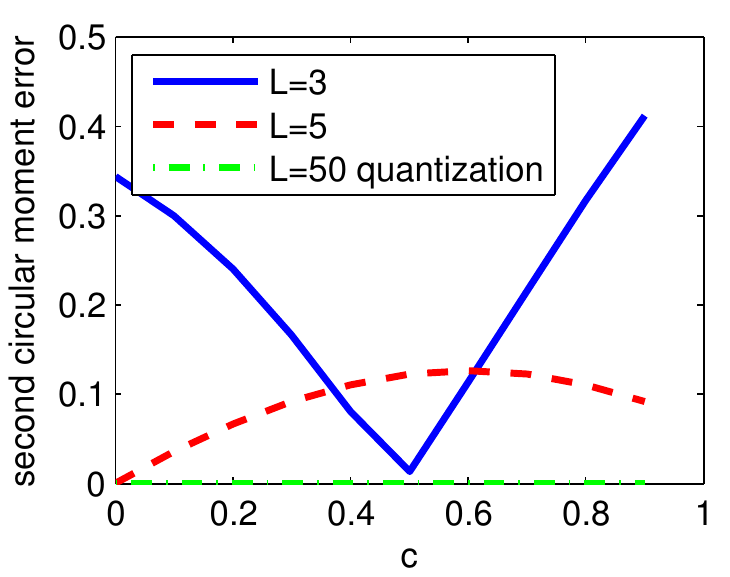} \\
	\includegraphics[width=\propagationwidth]{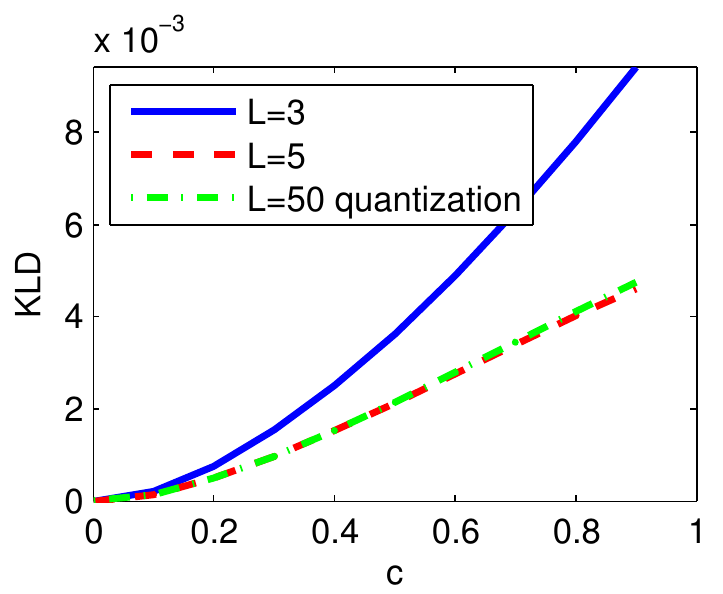} & \includegraphics[width=\propagationwidth]{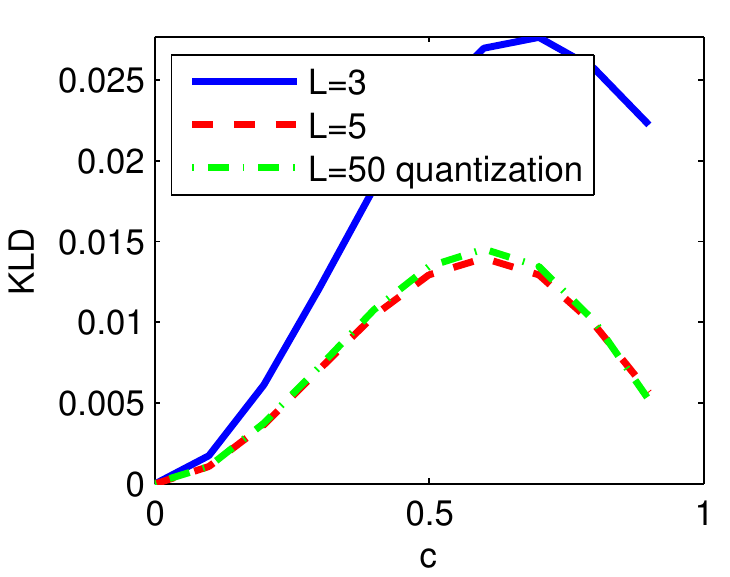} & 
	\includegraphics[width=\propagationwidth]{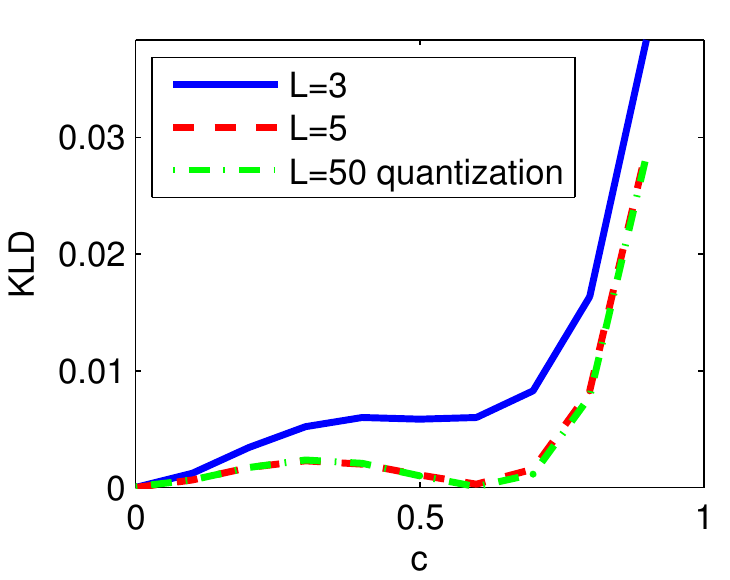} \\
	\end{tabular}
	\caption{Propagation of $\WN(0, \sigma)$ through a nonlinear function with nonlinearity parameter $c$.}
	\label{fig:propagationresults}
\end{figure*}

In order to evaluate the deterministic sampling as introduced in Sec.~\ref{sec:deterministic}, we investigate the accuracy when performing propagation through the nonlinear function $g: \Sone \to \Sone$, 
\begin{align*}
g(x) = x + c\times_\mathbb{R} \sin(x) \mod 2 \pi \ ,
\end{align*}
where $c\in [0,1)$ is a parameter controlling the strength of the nonlinearity and $\times_\mathbb{R}$ refers to multiplication in the field of real numbers $\mathbb{R}$. Furthermore, we consider the density $\WN(0, \sigma)$ that we want to propagate through $g(\cdot)$. For this purpose, we sample $\WN(0, \sigma)$ deterministically using the methods described in Sec.~\ref{sec:deterministic} and obtain $\WD(\gamma_1, \dots, \gamma_L, \beta_1, \dots, \beta_L)$. Then, we apply $g(\cdot)$ componentwise, which yields $\WD(\gamma_1, \dots, \gamma_L, g(\beta_1), \dots, g(\beta_L))$. 

The true posterior is given by 
\begin{align*}
f^{true}(x) = \frac{f (g^{-1}(x); \mu, \sigma)}{g'(x)}
\end{align*}
and can only be calculated numerically. We evaluate the first and the second circular moment $m^{WD}_i, \ i=1,2$ of the resulting WD distribution and compare to the first and the second circular moment $m^{true}_i, \ i=1,2$ of the true posterior, which is obtained by numerical integration\footnote{Numerical integration produces very accurate results in this case, but is too slow for use in practical filtering applications.}. The considered error measure is given by $|m^{WD}_i - m^{true}_i|, \  i=1,2$, where $|\cdot|$ is the Euclidean norm in the complex plane. Additionally, we fit a WN density to the posterior WD by circular moment matching and numerically calculate the Kullback-Leibler divergence
\begin{align}
\label{eq:kld}
D_{\text{KL}} \left(f^{true} || f^{fitted} \right) = \int_0^{2 \pi} f^{true}(x) \log \left( \frac{f^{true}(x)}{f^{fitted}(x)} \right) \di x \ ,
\end{align}
between the true posterior and the fitted WN. The Kullback-Leibler divergence is an information theoretic measure to quantify the information loss when approximating $f^{true}$ with $f^{fitted}$. This concept is illustrated in Fig.~\ref{fig:propagationidea}.

The results for different values of $\sigma$ are depicted in Fig.~\ref{fig:propagationresults}. We compare several samplers, the analytic methods with $L=3$ components (Algorithm~\ref{algo:3diracs}) and $L=5$ components (Algorithm~\ref{algo:5diracs}, parameter $\lambda=0.5$) from Sec.~\ref{sec:deterministic:analytic} as well as the quantization approach discussed in Sec.~\ref{sec:deterministic:optimization}. It can be seen that the analytic solution with $L=5$ components is significantly better than the solution with $L=3$ components. The quantization-based solution is computationally quite demanding but gives almost optimal results. However, the analytic solution with $L=5$ components has comparable performance in terms of the Kullback-Leibler divergence even though the posterior moments are not calculated as accurately.

\subsection{Moment-Based WN Multiplication}
In this evaluation, we compare the two methods for WN multiplication given in Sec.~\ref{sec:operations:multiplication:wnvm} and Sec.~\ref{sec:operations:multiplication:wnmoment}. For two WN densities $\WN(\mu_1, \sigma_1)$ and $\WN(\mu_2, \sigma_2)$, we calculate the true product $f^{true} = \WN(\mu_1, \sigma_1) \cdot \WN(\mu_2, \sigma_2)$ and compare it to the WN approximation $f^{fitted}$.

In order to determine the approximation quality, we compute the Kullback-Leibler divergence (\ref{eq:kld}). Furthermore, we consider the $L^2$ distance
\begin{align*}
D_{L^2}\left( f^{true}, f^{fitted} \right) =  \int_0^{2 \pi} \left( f^{true}(x) - f^{fitted}(x) \right)^2 \di x \ .
\end{align*}
The results for different values of $\sigma_1, \mu_2,$ and $\sigma_2$ are depicted in Fig.~\ref{fig:wnmulkld} and Fig.~\ref{fig:wnmull2}. We keep $\mu_1$ fixed because only the difference between $\mu_2$ and $\mu_1$ affects the result. Multiplication is commutative, so we consider different sets of values for $\sigma_1$ and $\sigma_2$ to avoid redundant plots. 

As can be seen, the moment-based approach derived in Sec.~\ref{sec:operations:multiplication:wnmoment} significantly outperforms the approach from Sec.~\ref{sec:operations:multiplication:wnvm} in almost all cases according to both distance measures. The new approach is particularly superior for small uncertainties.

\newcommand{\kldwidth}{42mm}
\begin{figure*}[htb]
	\centering
	\begin{tabular}{cccccc}
	& $\sigma_1 = 0.1$ & $\sigma_1 = 0.4$ & $\sigma_1 = 1.0$ \\
	\rotatebox{90}{\qquad $\sigma_2 = 0.2$} & \includegraphics[width=\kldwidth]{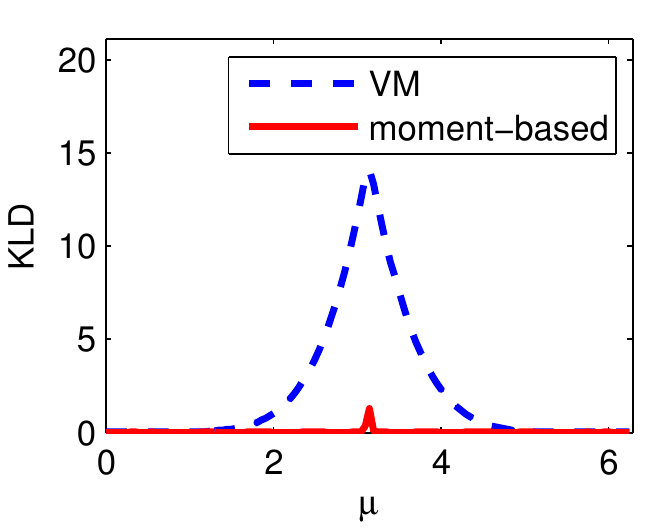} &  \includegraphics[width=\kldwidth]{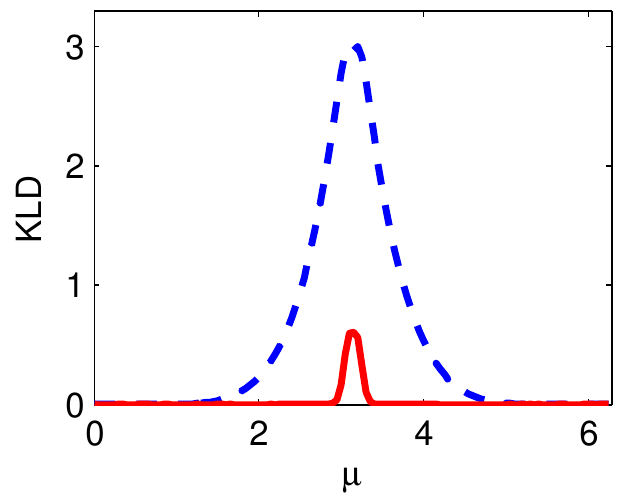} &  
	\includegraphics[width=\kldwidth]{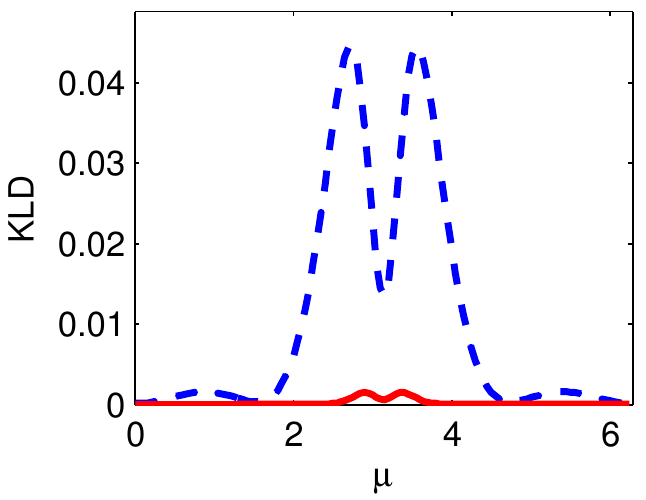} \\
	\rotatebox{90}{\qquad $\sigma_2 = 0.5$} & \includegraphics[width=\kldwidth]{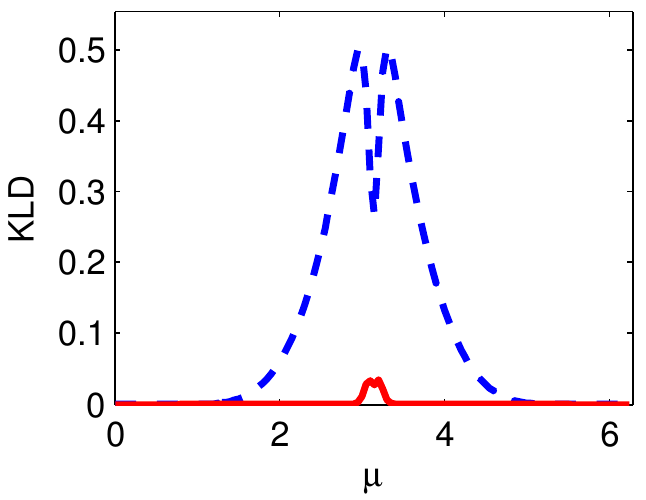} &  \includegraphics[width=\kldwidth]{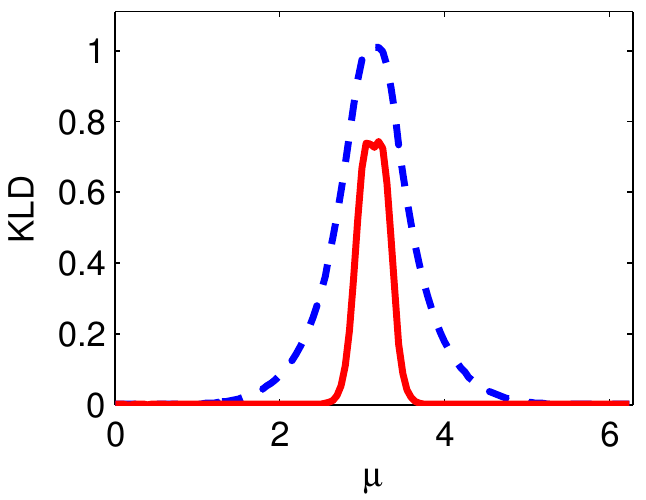} &  
	\includegraphics[width=\kldwidth]{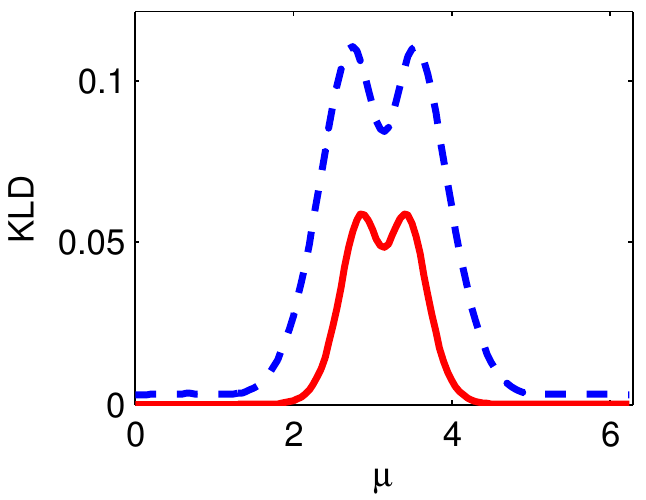} \\
	\rotatebox{90}{\qquad $\sigma_2 = 1.0$} & \includegraphics[width=\kldwidth]{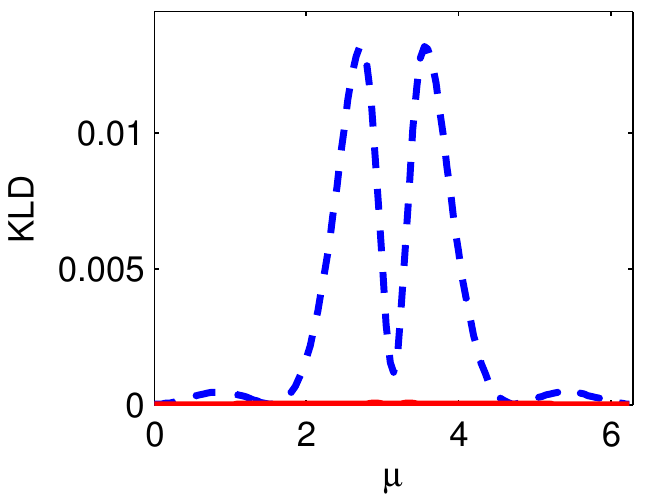} &  \includegraphics[width=\kldwidth]{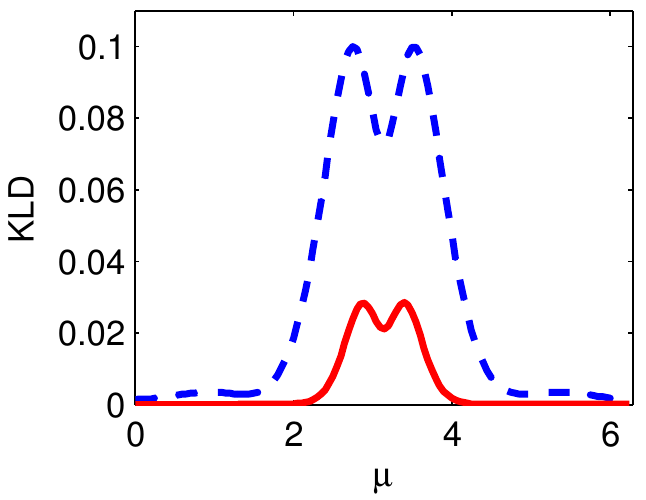} &  
	\includegraphics[width=\kldwidth]{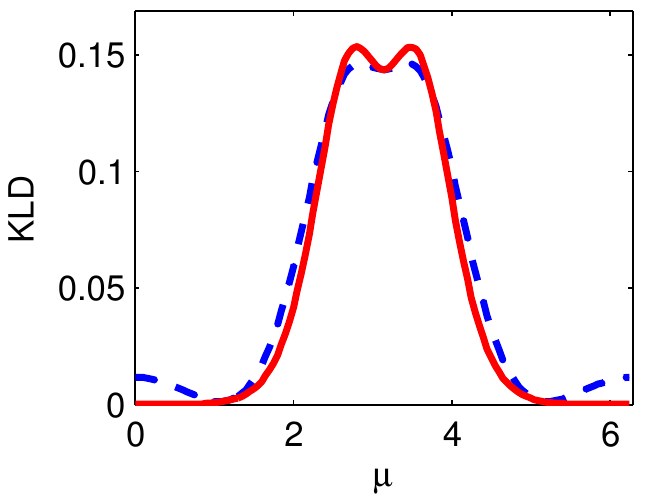} \\
	\end{tabular}
	\caption{Kullback-Leibler divergence between the true product of WN densities and the proposed approximations.}
	\label{fig:wnmulkld}
\end{figure*}

\begin{figure*}[htb]
	\centering
	\begin{tabular}{ccccccc}
	& $\sigma_1 = 0.1$ & $\sigma_1 = 0.4$ & $\sigma_1 = 1.0$ \\
	\rotatebox{90}{\qquad $\sigma_2 = 0.2$} & \includegraphics[width=\kldwidth]{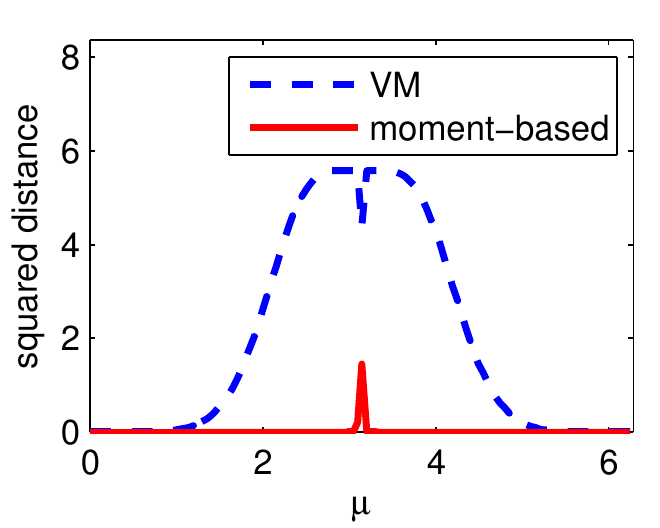} &  \includegraphics[width=\kldwidth]{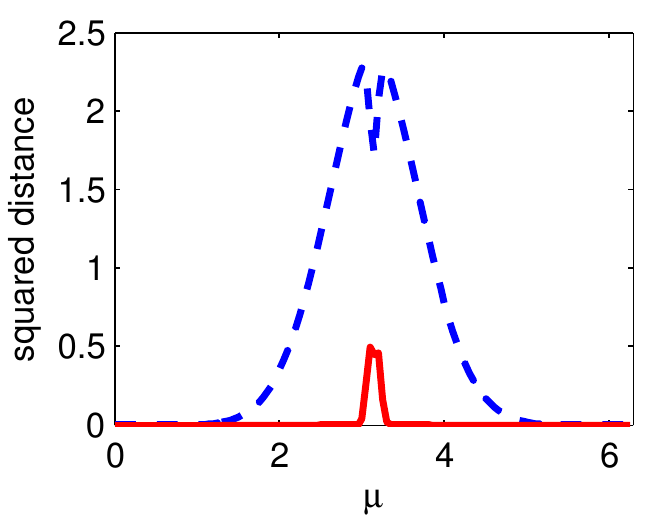} &  
	\includegraphics[width=\kldwidth]{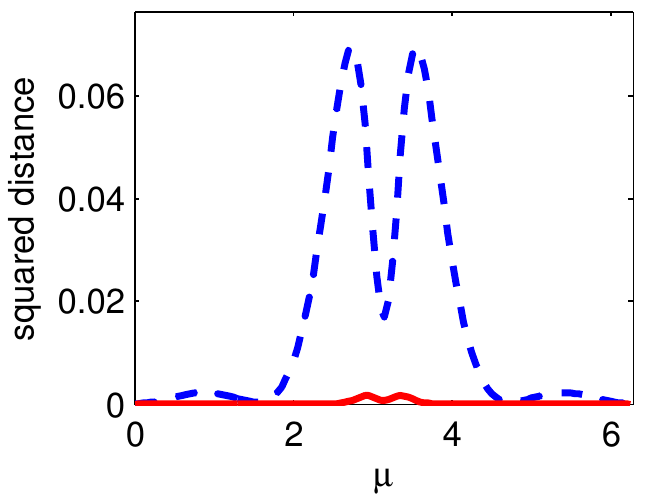} \\
	\rotatebox{90}{\qquad $\sigma_2 = 0.5$} & \includegraphics[width=\kldwidth]{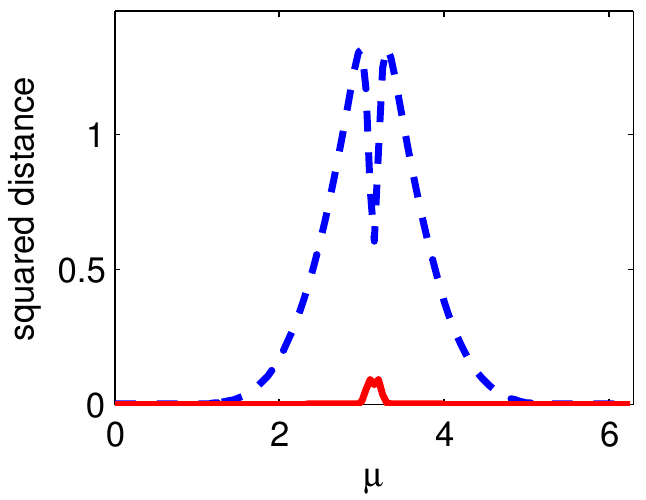} &  \includegraphics[width=\kldwidth]{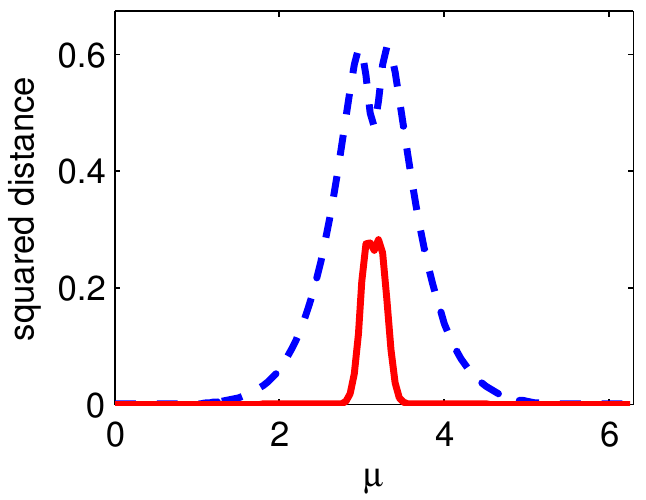} &  
	\includegraphics[width=\kldwidth]{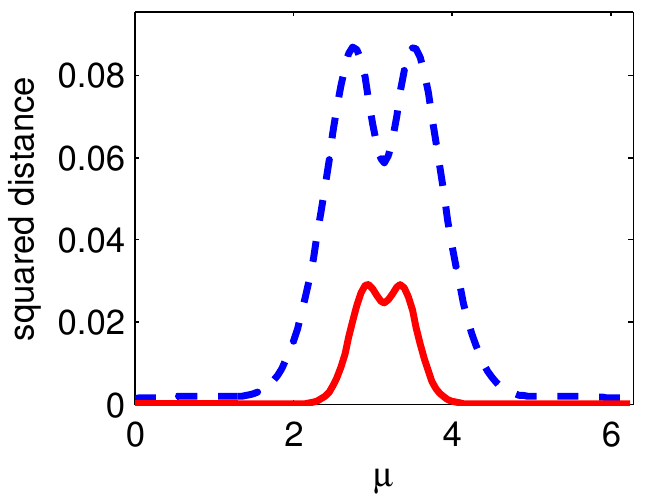} \\
	\rotatebox{90}{\qquad $\sigma_2 = 1.0$} & \includegraphics[width=\kldwidth]{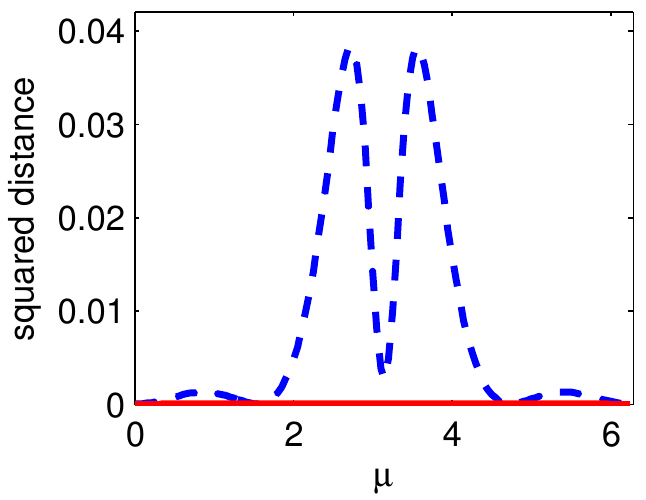} &  \includegraphics[width=\kldwidth]{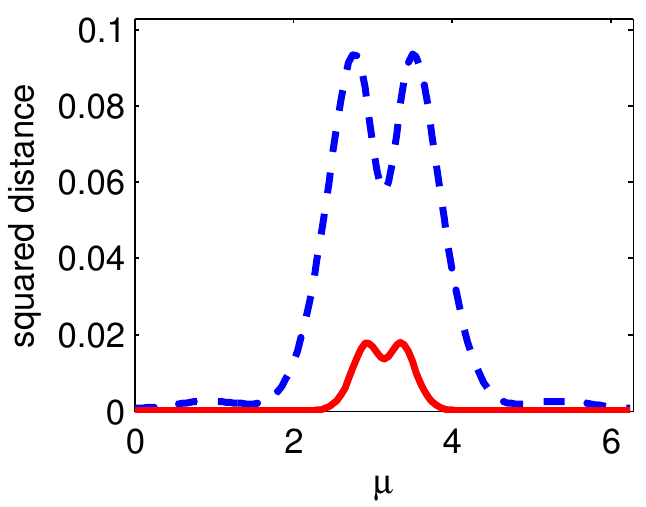} &  
	\includegraphics[width=\kldwidth]{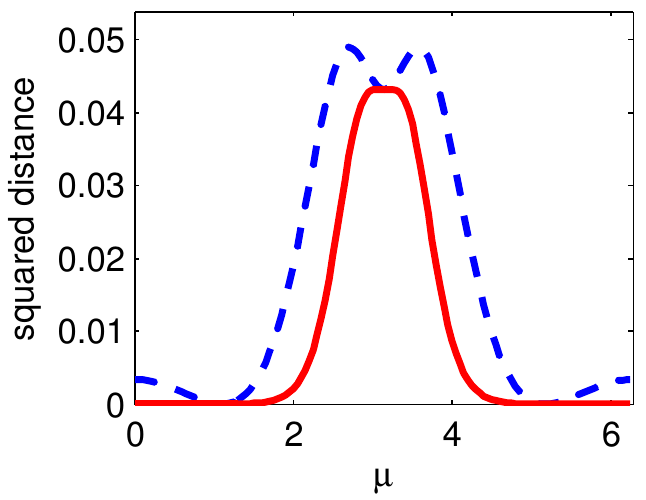} \\
	\end{tabular}
	\caption{$L^2$ distance between the true product of WN densities and the proposed approximations.}
	\label{fig:wnmull2}
\end{figure*}

\subsection{Filtering}

\begin{table}[htb]
	\centering
	\begin{tabular}{lcr}
	\toprule
	\textbf{Scenario} & \textbf{System Function} & \textbf{Measurement Noise $\mat{C}^v$} \\
	\midrule
	s & (\ref{eq:additive}) & $0.01 \cdot \matI_{2 \times 2}$ \\
	m & (\ref{eq:additive}) & $0.1 \cdot \matI_{2 \times 2}$ \\
	l & (\ref{eq:additive}) & $3\cdot \matI_{2 \times 2}$ \\
	s-non-additive & (\ref{eq:nonadditive}) & $0.01 \cdot \matI_{2 \times 2}$ \\
	m-non-additive & (\ref{eq:nonadditive}) & $0.1 \cdot \matI_{2 \times 2}$ \\
	l-non-additive & (\ref{eq:nonadditive}) & $3\cdot \matI_{2 \times 2}$ \\
	\bottomrule
	\end{tabular}
	\caption{Evaluation scenarios.}
	\label{table:scenarios}
\end{table}

In order to evaluate the proposed filtering algorithms, we simulated several scenarios. First of all, we distinguish between models with additive and with a more complex noise structure. In the case of additive noise, we consider the system function 
\begin{align}
x_{k+1} = x_k + c_1 \times_\mathbb{R} \sin(x_k) + c_2 + w_k \ ,
\label{eq:additive}
\end{align}
with two parameters $c_1 = 0.1, c_2=0.15$, noise $w_k \sim\WN(0,0.2)$, and $\times_\mathbb{R}$ is multiplication in the field of real numbers $\mathbb{R}$. Intuitively, $c_1$ determines the degree of nonlinearity and $c_2$ is a constant angular velocity that is added at each time step. For the case of arbitrary noise, the system function is given by
\begin{align}
x_{k+1} = x_k + c_1 \times_\mathbb{R} \sin(x_k + w_k) + c_2 \ ,
\label{eq:nonadditive}
\end{align}
with the same $c_1, c_2,$ and $w_k$ as above. In both cases, the nonlinear measurement function is given by 
\begin{align*} 
\hat{\vecz}_k = [\cos(x_k), \sin(x_k)]^T + \vecv_k \quad \in \mathbb{R}^2
\end{align*}
with additive noise $\vecv_k \sim \mathcal{N} (\vec{0}, \eta \cdot \matI_{2\times 2})$, $\eta \in \{ 3, 0.1, 0.01 \}$. An overview of all considered scenarios is given in Table~\ref{table:scenarios}.

In the scenarios with additive system noise, we compare the proposed filter to all standard approaches described in  Sec.~\ref{sec:problemformulation:standardapproaches}, a UKF with 1D state vector, a UKF with 2D state vector and particle filters with 10 and 100 particles. In order to handle non-additive noise with a UKF, typically state augmentation is used, which is not applicable to arbitrary noise. For this reason, we only compare the proposed approach to the particle filters in the non-additive noise case.

The initial estimate is given by $x_0 \sim \WN(0, 1)$, whereas the true initial state is on the opposite side of the circle $x_0^\text{true} = \pi$, i.e., the initial estimate is poor, which is difficult to handle for noncircular filters.

For the circular filtering algorithm, we use the deterministic sampling method given in Algorithm~\ref{algo:5diracs} with parameter $\lambda = 0.5$. The progression threshold is chosen as $R=0.2$.

In order to evaluate the performance of different filters, we consider a specific error measure that takes periodicity into account. The angular error is defined as the shortest distance on the circle
\begin{align*}
d: \Sone \times \Sone \to [0,\pi] \ , \ d(a,b) = \min (|a-b| ,2\pi - |a-b|) \ .
\end{align*}
This leads to an angular version of the commonly used root mean square error (RMSE)
\begin{align*}
\sqrt{\frac{1}{k_\text{max}} \sum_{k=1}^{k_\text{max}} d(x_k, x_k^\text{true})^2 }
\end{align*}
between estimates $x_k$ and true state variables $x_k^\text{true}$. We simulated the system for $k_\text{max}=100$ time steps and compared the angular RMSE of all estimators. The results from 100 Monte Carlo runs are depicted in Fig.~\ref{fig:filteringrmse}.

In the scenarios with additive noise, it can be seen that the proposed filter performs very well regardless of the amount of noise. Only the particle filter with 100 particles is able to produce similar results. However, it should be noted that the proposed filter uses just five samples. The particle filter with 10 particles performs a lot worse and fails completely for small noise as a result of particle degeneration issues. Both variants of the UKF perform worse than the proposed filter. Particularly the UKF with two-dimensional state does not work very well, which can be explained by the inaccuracies in the conversion of a one-dimensional into a two-dimensional noise noise.

When non-additive noise is considered, the proposed filter even significantly outperforms the particle filter with 100 particles. As a result of the low number of particles and the associated issues regarding particle degeneration, the particle filter with 10 particles has the worst performance.

\begin{figure*}
	\centering
	\begin{tabular}{ccc}
	small noise $\eta=0.01$ & medium noise $\eta = 0.1$ & large noise $\eta = 3$ \\
	\includegraphics[width=5cm]{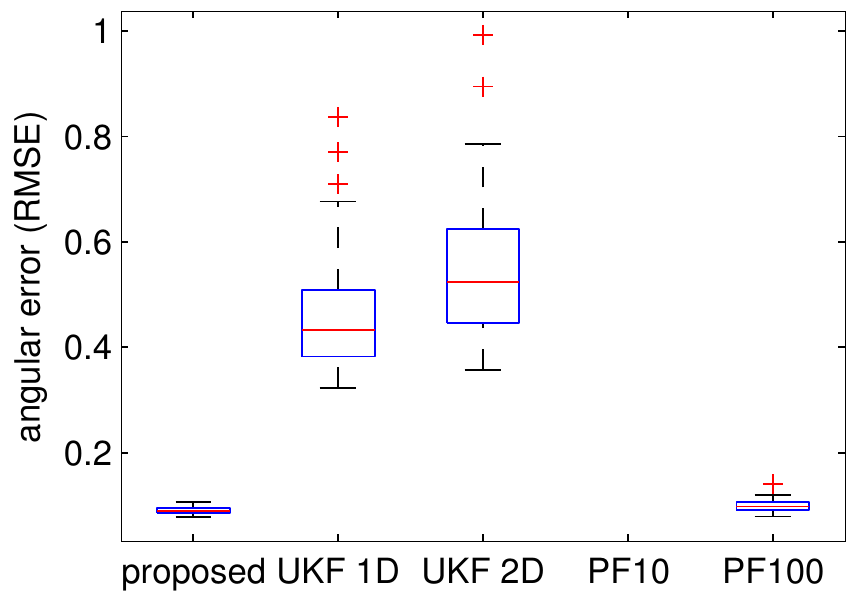} &
	\includegraphics[width=5cm]{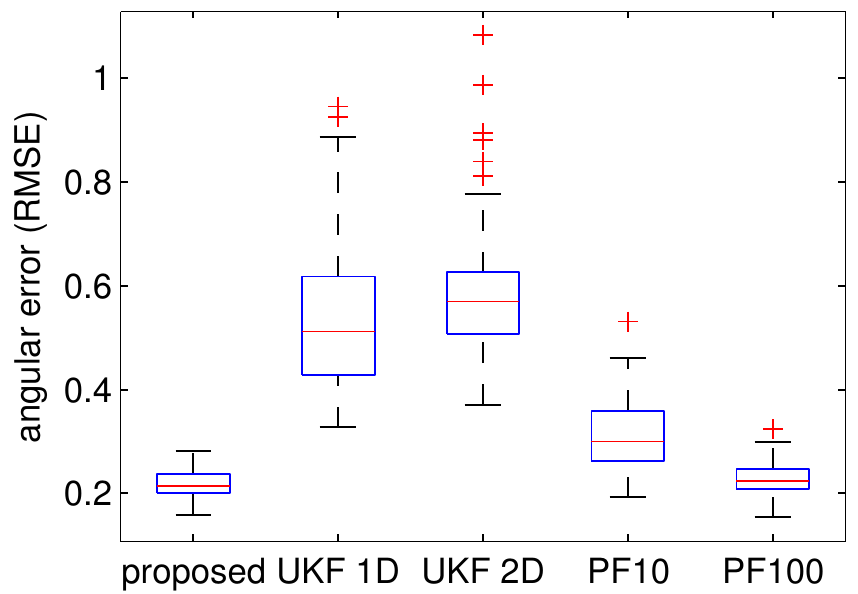} &
	\includegraphics[width=5cm]{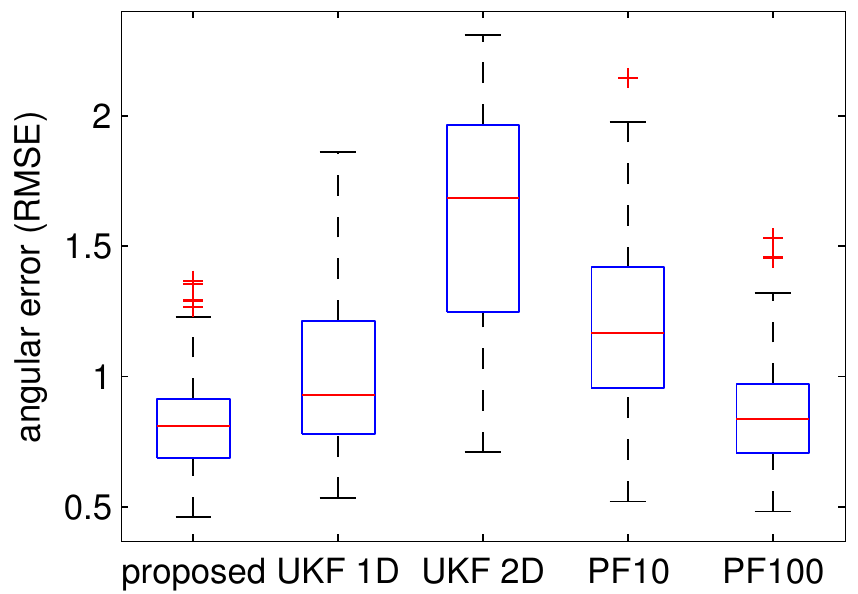} \\
	\includegraphics[width=5cm]{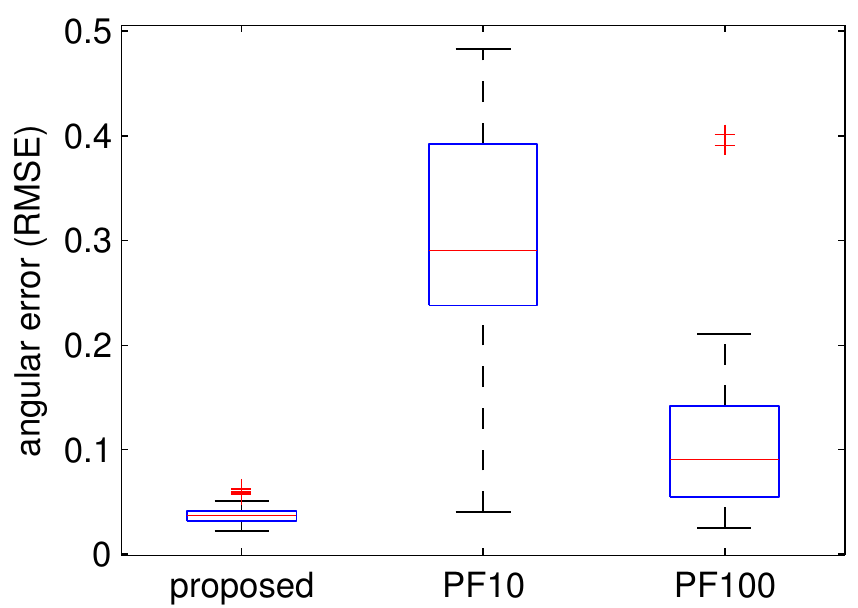} &
	\includegraphics[width=5cm]{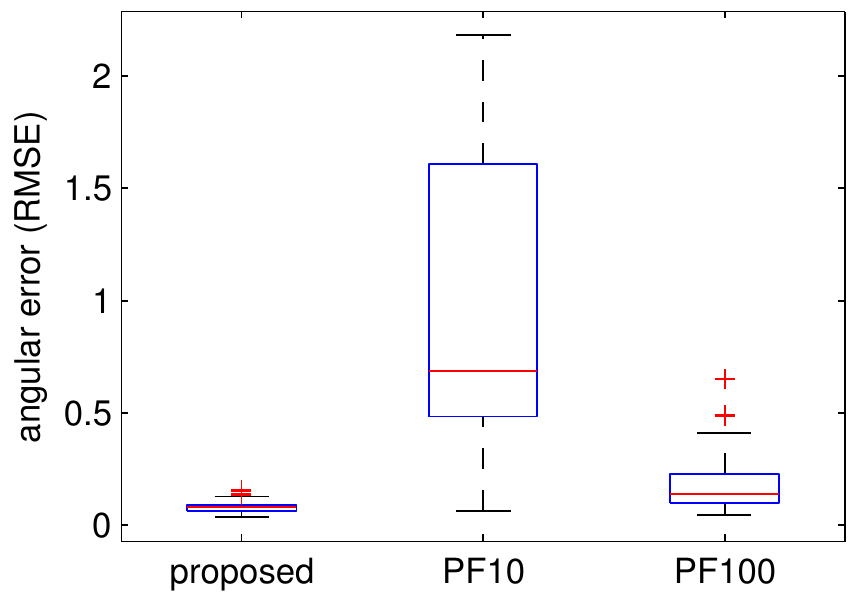} &
	\includegraphics[width=5cm]{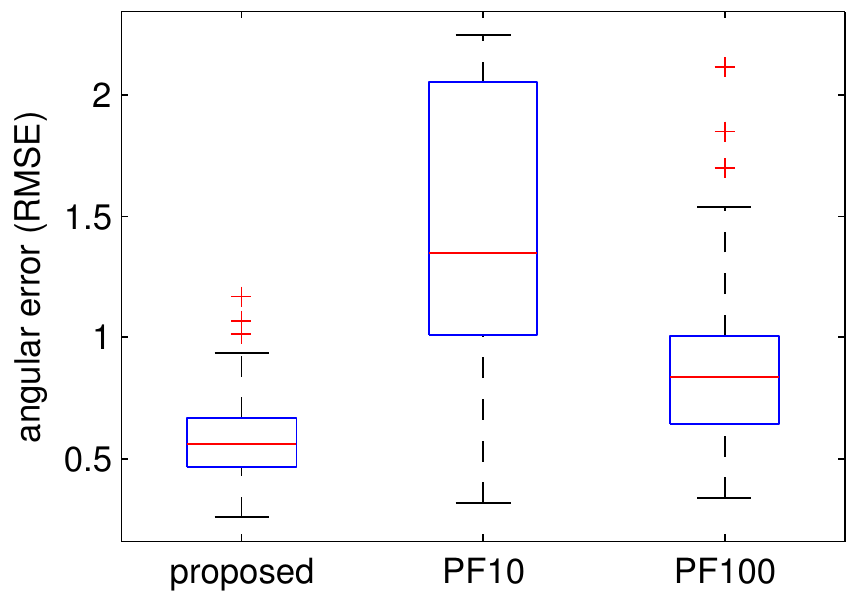} \\
	\end{tabular}
	\caption{RMSE (in radians) for different filters obtained from 100 Monte Carlo runs for additive noise (top) and non-additive noise (bottom).}
	\label{fig:filteringrmse}
\end{figure*}

\section{Conclusion}
\label{sec:conclusion}
In this paper, we presented a framework for recursive filtering on the circle. The proposed filtering algorithms can deal with arbitrary nonlinear system and measurement functions. Furthermore, they can be used in conjunction with different circular probability distributions. We have shown that the prediction step can be performed based on circular moments only, without ever assuming a particular distribution.

For the purpose of evaluation, we have considered several aspects of the proposed methods. First of all, the accuracy of deterministic approximations was evaluated by considering the error when using them to propagate a continous distribution through a nonlinear function. We have found that the proposed deterministic approximation with five samples yields good results for most practical scenarios. Second, we evaluated the novel moment-based WN multiplication method and show that it is superior to the previously published method based on fitting a VM distribution. Finally, we evaluated the proposed filtering algorithms in several scenarios and compared it to standard approaches. These simulations show the advantages of using a circular filtering scheme compared to traditional methods intended for the linear case.

Future work may include extensions of the proposed methods to other manifolds such as the torus or the hypersphere. Additionally, consideration of multimodal circular distributions may be of interest, for example by means of WN or VM mixtures.

\section*{Acknowledgment}
This work was partially supported by grants from the German Research Foundation
(DFG) within the Research Training Groups RTG 1126 ``Soft-tissue Surgery: New
Computer-based Methods for the Future Workplace'' and RTG 1194 ``Self-organizing
Sensor-Actuator-Networks''.

\bibliographystyle{IEEEtran_Capitalize}
\bibliography{../BibTeX/ISASPreprints,../BibTeX/ISASPublikationen,../BibTeX/ISASPublikationen_laufend,../../../Literatur/gk,unpublished}

\end{document}